\documentclass{CSML}

\def\dOi{13(1:12)2017}
\lmcsheading%
{\dOi}
{1--35}
{}
{}
{Apr.~29, 2016}
{Mar.~22, 2017}
{}

\usepackage{hyperref}

\def\glet{\global\let}

\theoremstyle{plain}
\usepackage{t1enc}
\usepackage[babel]{csquotes}

\usepackage{alltt}
\usepackage{subfig}
\usepackage{graphicx}
\let\subfigure\subfloat

\usepackage{latexsym}
\usepackage[all]{xy}
\usepackage[curve]{xypic}
\usepackage{tikz}
\usetikzlibrary{trees}
\usetikzlibrary{arrows}
\usepackage{listings}
\usepackage{amssymb}
\usepackage{ysmath}

\def\TT{\textit{tt}}
\def\FF{\textit{ff}}
\def\timFont{\bf}
\renewcommand{\ie}{{\it i.e.~}}
\newcommand{\Ops}{{\timFont Ops}}
\newcommand{\Vars}{{\timFont Vars}}
\newcommand{\TRS}{{\timFont TRS}}
\newcommand{\Automaton}{{\timFont Automaton}}

\newcommand{\States}{{\timFont States}}
\newcommand{\FinalStates}{{\timFont Final States}}
\newcommand{\Transitions}{{\timFont Transitions}}
\newcommand{\Equations}{{\timFont Equations}}
\newcommand{\Rules}{{\timFont Rules}}

\definecolor{OliveGreen}{rgb}{.0, 0.36, 0.16}

\def\D{\mathcal D}
\def\C{\mathcal C}

\def\ensfonc{\Sigma}
\def\ensvar{\mathcal X}
\def\signature{(\ensfonc,\ensvar)}
\def\sfa#1#2{{\vphantom{#2}_{#1}\kern-.05em #2}}
\def\langtermesaux(#1,#2){T(#1, #2)}
\def\langtermes#1{\expandafter\langtermesaux#1}
\def\langtermesclos#1{T(#1)}

\DeclareMathOperator\pos{Pos}
\def\ssterme#1#2{{#1}_{|#2}}
\def\rempterme#1#2#3{{#1}{\left[#3\right]}_{#2}}

\def\automate{\mathcal A}
\def\bautomate{\mathcal B}
\def\langrec#1{\mathop{\raisebox{.15pt}{$\mathscr L$}}(#1)}
\def\langrecs#1#2{\mathop{\raisebox{.15pt}{$\mathscr L$}}(#1, #2)}

\def\move{\mathrel{\mathop\rightarrowtail}}
\def\PC{\mathcal{C\kern-0.4ex P}}
\def\PCx#1{\PC\kern-0.4ex_{#1}}

\def\automatefixin{{\automate_{\mathrm{in}*}}}
\def\afixrin{{\automate_{\mathrm{rin}*}}}

\newcommand{\autoInit}{\automate_{init}}

\def\ctxttriv{\square}
\def\Rin{R_{\mathrm{in}}}
\def\Rrin{R_{\mathrm{rin}}}

\DeclareMathOperator{\lirr}{I\textsc{rr}}
\DeclareMathOperator\airr{\mathcal{A}\kern-.0ex\mathchoice{\scriptstyle}{\scriptstyle}{\scriptscriptstyle}{\scriptscriptstyle}\mathcal{I\kern-.2exR\kern-.4exR}}
\def\pred{p_{\mathrm{red}}}

\def\paire(#1,#2){\left< #1, #2\right>}
\def\projgauche#1{\mathop{\Pi_1}\left(#1\right)}
\def\projdroite#1{\mathop{\Pi_2}\left(#1\right)}

\def\colour#1{\mathfrak{#1}}
\def\asc#1#2{{#1}^{\notcolour{#2}}}
\def\notcolour#1{\cancel {\mathfrak{#1}}}
\def\classcol#1#2{\left[#2\right]_{\colour{#1}}}
\def\movecol#1{\mathrel{\mathop\rightarrowtail\limits^{\colour{#1}}}}
\def\sseps{^{\ssepsaux}}
\def\ssepsaux{\notcolour R}
\def\moveun#1{\mathrel{\mathop\rightarrowtail\limits_{#1}}}
\def\moveuncol#1#2{\mathrel{\mathop\rightarrowtail\limits_{#1}^{\colour{#2}}}}

\def\moveet#1{\mathrel{\mathop\rightarrowtail\limits^{\ast}_{#1}}}
\def\moveetcol#1#2{\mathrel{\mathop\rightarrowtail\limits^{\colour{#2},\ast}_{#1}}}

\def\moverepcol#1#2{\mathrel{\mathop{\ooalign{\raisebox{.375ex}{$\rightarrowtail$}\cr\raisebox{-.375ex}{$\leftarrowtail$}}}\limits_{#1}^{\colour{#2}}}}
\def\notmoverepcol#1#2{\mathrel{\cancel{\mathop{\ooalign{\raisebox{.375ex}{$\rightarrowtail$}\cr\raisebox{-.375ex}{$\leftarrowtail$}}}\limits_{#1}^{\colour{#2}}}}}

\def\moverepetcol#1#2{\mathrel{\mathop{\ooalign{\raisebox{.375ex}{$\rightarrowtail$}\cr\raisebox{-.375ex}{$\leftarrowtail$}}}\limits^{\colour{#2}, \ast}_{#1}}}

\def\quotrep#1{#1\kern-.25ex/\kern-.35ex{\mathord{\ooalign{\raisebox{.375ex}{\tiny$\rightarrowtail$}\cr\raisebox{-.2ex}{\tiny$\leftarrowtail$}}}}}

\def\quoteqv#1#2{#1\kern-.1ex/\kern-.35ex{\mathord{\raisebox{-.2ex}{\tiny $#2$}}}}
\def\quotcol#1#2{#1\kern-.1ex/\kern-.35ex{\mathord{\raisebox{-.2ex}{$\mathchoice{\scriptstyle}{\scriptstyle}{\scriptscriptstyle}{\scriptscriptstyle}\colour{#2}$}}}}

\def\normalisation#1{\mathop{N\kern-1pt o\kern-1pt r\kern-1pt m_{#1}}}
\def\normalisationaux#1#2{\mathop{N\kern-1.3pt o\kern-1.3pt r\kern-1.3pt m\kern-1pt A\kern-1.3pt u\kern-1.3pt x_{#1}^{#2}}}

\newcommand{\timbuk}{{\sf Timbuk}}

\def\Q{Q}
\def\F{\ensfonc}
\def\X{\ensvar}
\def\R{R}
\def\rw{\rightarrow}
\def\rwR{\rightarrow_R}

\def\TFQ{\langtermesclos{\ensfonc,{\Q}}}
\def\TF{\langtermesclos\ensfonc}

\def\TC{\langtermesclos\C}
\def\sep{\: | \:}

\def\A{\automate}
\def\B{\bautomate}

\def\Lange(#1){\langrec{{#1}\sseps}}
\def\Lang(#1){\langrec{#1}} 

\makeatletter
\newenvironment{enumerate-ligne}{\par\setcounter{enumi}{0}\def\item{\ifhmode \unskip\hfill\else\leavevmode\hskip 0pt plus .45fill\fi\refstepcounter{enumi}\sbox \@tempboxa {\@mklab {\labelenumi}}\global \setbox \@labels \hbox {\unhbox \@labels  \hskip -\labelwidth \hskip -\labelsep \ifdim \wd \@tempboxa >\labelwidth \box \@tempboxa \else \hbox to\labelwidth {\unhbox \@tempboxa }\fi \hskip \labelsep }\box\@labels\ignorespaces}}{\hskip 0pt plus .45fill\hskip\labelwidth\hskip\labelsep\null}
\def\anditem{\unskip\hskip -.5em plus.5fill{and}\hskip 0pt plus-.5fill\null \item}
 \def\anditem{ and\item}
\makeatother

\begin{document}

\title[Reachability Analysis of Innermost Rewriting]{Reachability Analysis of
  Innermost Rewriting\rsuper*}
\titlecomment{{\lsuper*}This paper is an extended version of the paper~\cite{GenetS-RTA15}.}

\author[T.~Genet]{Thomas Genet}	
\address{IRISA, Campus de Beaulieu, 35042 Rennes Cedex, France}	
\email{\{Thomas.Genet, Yann.Salmon\}@irisa.fr}  

\author[Y. Salmon]{Yann Salmon}	
\address{\vspace{-18 pt}}	




\begin{abstract}
\noindent
We consider the problem of inferring a grammar describing the output of a
functional program given a grammar describing its input.  Solutions to this
problem are helpful for detecting bugs or proving safety properties of
functional programs, and several rewriting tools exist for solving this problem.
However, known grammar inference techniques are not able to take evaluation
strategies of the program into account. This yields very imprecise results when
the evaluation strategy matters. In this work, we adapt the Tree Automata
Completion algorithm to approximate accurately the set of terms reachable by
rewriting under the innermost strategy. We formally prove that the proposed
technique is sound and precise w.r.t. innermost rewriting. We show that those
results can be extended to the leftmost and rightmost innermost case. The
algorithms for the general innermost case have been implemented in the Timbuk
reachability tool. Experiments show that it noticeably improves the accuracy of
static analysis for functional programs using the call-by-value evaluation strategy.
\end{abstract}

\maketitle

\section{Introduction and motivations}
If we define by a grammar the set of inputs of a functional program, is it
possible to infer the grammar of its output? 
Some strongly typed functional programming languages (like Haskell, OCaml, Scala and F\#)
have a type inference mechanism. This mechanism, among others, permits to
automatically detect some kinds of errors in the programs. In particular, when the
inferred type is not the expected one, this suggests that there may be a bug in
the function. To prove properties stronger than well typing of a program, it is
possible to define properties and, then, to prove them using a proof assistant
or an automatic theorem prover. However, defining those properties with logic
formulas (and do the proof) generally requires a strong expertise.

Here, we focus on a restricted family of properties: regular properties on the
structures manipulated by those programs. Using a grammar, we define the set of
data structures given as input to a function and we want to infer the grammar
that can be obtained as output (or an approximation). Like in the case of type
inference, the output grammar can suggest that the program contains a bug, or on
the opposite, that it satisfies a regular property.

The family of properties that can be shown in this way is restricted, but it
strictly generalizes standard typing as used in languages of the ML family\footnote{Standard types can easily be expressed
  as grammars. The opposite is not true. For instance, with a grammar one can
  distinguish between an empty and a non empty list.}.
There are other approaches where the type system is enriched by logic formulas
and arithmetic like~\cite{VazouRJ-ESOP2013,CastagnaNXHLP-POPL14},
but they generally require to annotate the output of the function for type
checking to succeed. The properties we consider here are 
intentionally simpler so as to limit as much as possible the need for
annotations. The objective is to define a {\em lightweight} formal verification
technique. The verification is {\em formal} because it {\em proves} that the
results have a particular form. But, the verification is {\em lightweight} for
two reasons. First, the proof is carried out automatically: no interaction with
a prover or a proof assistant is necessary. Second, it is not necessary to state
the property on the output of the function using complex logic formulas or an
enriched type system but, instead, only to observe and check the result of an
abstract computation.

With regards to the grammar inference technique itself, many works are devoted
to this topic in the functional programming
community~\cite{JonesA-TCS07,KobayashiTU-POPL10,OngR-POPL11}\footnote{Note that the objective
  of other papers like~\cite{BroadbentCHS-ICFP13,Kobayashi-ACM13} is different. They aim at
  predicting the control flow of a program rather than estimating the possible
  results of a function (data flow).} as well as in the
rewriting
community~\cite{Genet-RTA98,TakaiKS-RTA00,
BoichutCHK-IJFCS09,GenetR-JSC10,KochemsO-RTA11,Lisitsa-RTA12,BoichutCR-RTA13,Genet-WRLA14,Genet-JLAMP15}.  
In~\cite{Genet-WRLA14,Genet-JLAMP15}, starting from a term rewriting system (TRS for short) encoding a function and
a tree automaton recognising the inputs of a function, it is possible to
automatically produce a tree automaton {\em over-approximating} as precisely as
possible the outputs. Note that a similar reasoning can be done on higher-order
programs~\cite{JonesA-TCS07,GenetS-rep13} using a well-known encoding of higher order
functions into first-order TRS~\cite{Reynolds-IP69}. However, for the sake of 
simplicity, most examples used in this paper will be first order functions.
This is implemented in the \timbuk\ tool~\cite{timbuk}. Thus, we are 
close to building an {\em abstract interpreter}, evaluating a function on an
(unbounded) regular set of inputs, for a real programming language. However,
none of the aforementioned grammar inference techniques takes the evaluation
strategy into account, though every functional programming language has
one. As a consequence, those techniques produce very poor results as soon as the evaluation
strategy matters or, as we will see, as soon as the program is not terminating. 
This paper proposes a grammar inference technique for the innermost strategy:
\begin{itemize}
\item overcoming the precision problems of~\cite{JonesA-TCS07,OngR-POPL11}
and~\cite{Genet-RTA98,TakaiKS-RTA00,Takai-RTA04, 
BoichutCHK-IJFCS09,GenetR-JSC10,Lisitsa-RTA12,BoichutCR-RTA13,Genet-WRLA14,Genet-JLAMP15} on
the analysis of functional programs using call-by-value strategy
\item whose accuracy is not only shown on a practical point of view but also 
formally proved. This is another improvement w.r.t. all others grammar inference techniques
(except~\cite{GenetR-JSC10}).
\end{itemize}

\subsection{Towards an abstract OCaml interpreter}

In the following, we assume that we have an abstract OCaml interpreter. This
interpreter takes a regular expression as an input and outputs another regular
expression. In fact, all the computations
presented in this way have been performed with 
\timbuk\ (and latter with \timbuk STRAT), but on a TRS and a tree automaton
rather than on an OCaml function and a regular expression. We made this
choice to ease the understanding of input and output languages, since regular
expressions are far more easier to read and to understand than tree
automata. Assume that we have a notation, inspired by regular expressions, to
define regular languages of lists. Let us denote by {\tt [a*]} (resp. {\tt
  [a+]}) the language of lists having 0 (resp. 1) or more occurrences of symbol
{\tt a}.
We denote by {\tt [(a|b)*]} any list with 0 or more occurrences of {\tt a} and {\tt b} (in any
order). Now, in OCaml, we define a function deleting all the occurrences of an
element in a list. Here is a first (bugged) version of this function:

\begin{lstlisting}[escapechar=@,numbers=none]
let rec delete x ls= match ls with
  | [] -> []
  | h::t -> if h=x then t else h::(delete x t);;
\end{lstlisting}

\noindent
Of course, one can perform tests on this function using the usual OCaml interpreter:

\medskip
\noindent
{\small
\textcolor{OliveGreen}{\tt \# delete 2 [1;2;3];;} \\
\textcolor{OliveGreen}{\tt -:int list= [1;3]}
}

\medskip
\noindent
With an {\em abstract} OCaml interpreter dealing with grammars, we could
ask the following question: what is the set of the results obtained
by applying {\tt delete} to {\tt a} and to any list of
{\tt a} and {\tt b}?

\medskip
\noindent
{\small
\textcolor{OliveGreen}{\tt \# delete a [(a|b)*];;} \\
\textcolor{OliveGreen}{\tt -:abst list= [(a|b)*]}
}
\medskip

\noindent
The obtained result is not the expected one. Since all occurrences of {\tt a}
should have been removed, we expected the result {\tt [b*]}. Since the abstract
interpreter results into a grammar {\em over-approximating} the set of outputs, this
does not {\em show} that there is a bug, it only suggests it (like for type
inference). Indeed, in the definition of {\tt delete} there is a missing 
recursive call in the {\tt then} branch. If we correct this mistake, we get:

\medskip

\noindent
{\small
\textcolor{OliveGreen}{\tt \# delete a [(a|b)*];;}\\
\textcolor{OliveGreen}{\tt -:abst list= [b*]}
}

\medskip

\noindent
This result proves that {\tt delete} deletes all occurrences of an element in a
list. This is only one of the expected properties of {\tt delete}, but shown
automatically and without complex formalization. 
Here is, in \timbuk{} syntax, the TRS $R$ (representing {\tt delete}) and the
tree automaton $A0$ (representing {\tt [(a|b)*]}) that are given to
\timbuk{} to achieve the above proof.

\medskip
{\small
\begin{alltt}
\Ops delete:2 cons:2 nil:0 a:0 b:0 ite:3 true:0 false:0 eq:2
\Vars X Y Z

\TRS R
eq(a,a)->true          eq(a,b)->false       eq(b,a)->false       eq(b,b)->true               
delete(X,nil)->nil     ite(true,X,Y)->X     ite(false,X,Y)->Y
delete(X,cons(Y,Z))->ite(eq(X,Y),delete(X,Z),cons(Y,delete(X,Z)))

\Automaton A0  
\States qf qa qb qlb qlab qnil  
\FinalStates qf
\Transitions  delete(qa,qlab)->qf  a->qa  b->qb  nil->qlab
cons(qa,qlab)->qlab    cons(qb,qlab)->qlab
\end{alltt}
}

\medskip
\noindent
The resulting automaton computed by \timbuk\ is the following. It is not
minimal but its recognised language is equivalent to {\tt [b*]}.

\medskip
{\small 
\begin{alltt}
\States q0 q6 q8 
\FinalStates q6 
\Transitions cons(q8,q0)->q0  nil->q0  b->q8  cons(q8,q0)->q6  nil->q6
\end{alltt}
}  

\subsection{What is the problem with evaluation strategies?}

Let us consider the function {\tt sum(x)} which computes the sum
of the {\tt x} first natural numbers.

\medskip
{\small
\begin{lstlisting}[numbers=none]
let rec sumList x y=           let rec nth i (x::ls)= 
  (x+y)::(sumList (x+y) (y+1))    if i<=0 then x else nth (i-1) ls
let sum x= nth x (sumList 0 0)
\end{lstlisting}}

  This function is terminating with call-by-need (used in Haskell) but not with
  call-by-value strategy (used in OCaml). Hence, any call to {\tt sum} for any
  number {\tt i} will not terminate because of OCaml's evaluation strategy. Thus
  the result of the abstract interpreter on {\tt sum s*(0)} ({\em i.e.} {\tt
    sum} applied to any natural number {\tt 0}, {\tt s(0)}, \ldots) should be an
  empty grammar meaning that there is an empty set of results. However, if we
  use any of the techniques mentioned in the introduction to infer the output
  grammar, it will fail to show this. All those techniques compute reachable
  term grammars that do not take evaluation strategy into account. In
  particular, the inferred grammars will also contain all call-by-need
  evaluations. Thus, an abstract interpreter built on those techniques will
  produce a result of the form {\tt s*(0)}, which is a very rough
  approximation. In this paper, we propose to improve the accuracy of such
  approximations by defining a language inference technique taking the
  call-by-value evaluation strategy into account.

\subsection{Computing over-approximations of innermost reachable terms}

Call-by-value evaluation strategy of functional programs is strongly related to
innermost rewriting. The problem we are interested in is thus to compute (or to
over-approximate) the set of innermost reachable terms. 
For a TRS $R$ and a set of terms
$L_0\subseteq \langtermesclos{\ensfonc}$, the set of reachable terms is
$R^*(L_0)=\enstq{t \in \langtermesclos{\ensfonc}}{\exists s\in L_0, s \to^*_{R}
  t}$. This set can be computed for specific classes of $R$ but, in general, it
has to be approximated. 
Most of the techniques compute such approximations using tree automata (and not
grammars) as the core formalism to represent or approximate the (possibly)
infinite set of terms $R^*(L_0)$. Most 
of them also rely on a Knuth-Bendix completion-like algorithm to produce an
automaton $\automate^*$ recognising exactly, or over-approximating, the set of
reachable terms. As a result, these techniques can be referred to as {\em tree
  automata completion}
techniques~\cite{Genet-RTA98,TakaiKS-RTA00,BoichutCHK-IJFCS09,GenetR-JSC10,Lisitsa-RTA12}. 

Surprisingly, very little effort has been paid to computing or
over-approximating the set $R_{strat}^*(L_0)$, \ie set of reachable terms when $R$ is
applied with a strategy $strat$. To the best of our knowledge, Pierre R\'{e}ty
and Julie Vuotto's work~\cite{RetyV-RTA02} is the first one to have tackled this goal.  They give
some sufficient conditions on $L_0$ and $R$ for $R_{strat}^*(L_0)$ to be 
recognised by a tree automaton $\automate^*$, where $strat$ can be the innermost
or the outermost strategy. Innermost reachability for shallow TRSs was also
studied in~\cite{GodoyJacquemard-WRS08}. However, in both cases, the
restrictions on $R$ are  
strong and generally incompatible with functional programs seen as TRS. Moreover, the
proposed techniques are not able to over-approximate reachable terms when the TRSs
does not satisfy the restrictions.

In this paper, we concentrate on the innermost strategy and define a tree
automata completion algorithm over-approximating the set $\Rin^*(L_0)$
(innermost reachable terms) for any left-linear TRS $R$ and any regular set of
input terms $L_0$. As the completion algorithm of~\cite{GenetR-JSC10}, it is
parameterized by a set of term equations $E$ defining the precision of the
approximation. We prove the soundness of the algorithm: for all set of equation
$E$, if completion terminates then the resulting automaton $\automate^*$
recognises an over-approximation of $\Rin^*(L_0)$.  Then, we prove a precision
theorem: $\automate^*$ recognises no more terms than terms reachable by
innermost rewriting with $R$ modulo equations of $E$.  We also show how these
theorems can be extended to rightmost (or leftmost) innermost. Finally, we show on
several examples that, using innermost completion, we noticeably improve the
accuracy of the static analysis of functional programs.

This paper is an extended version of~\cite{GenetS-RTA15}. With regards to the
original paper, this paper contains the full proofs and the correctness and precision
theorems have been generalized to reachable and to irreducible reachable terms (normalized
forms). The completion technique and both correctness and precision theorems have been extended to the 
leftmost/rightmost innermost strategy. Finally, the paper includes several
detailed examples (including a higher-order one) that were not part of the
conference paper. This paper is organized as follows. Section~\ref{sect-notations} recalls some
basic notions about TRSs and tree automata.
Section~\ref{sect-innermost} exposes innermost completion.
Section~\ref{sect-th} states and proves the soundness of this method. Section~\ref{sect-precision} states the precision theorem.
Section~\ref{sect-fun} demonstrates how our new technique can effectively give more precise results on functional programs thanks to the tool \timbuk STRAT, an implementation of our method in the \timbuk{} reachability tool~\cite{timbuk}.
Section~\ref{equation} explains how equations can be inferred from the TRS to
analyze. Section~\ref{right} presents a direct extension of the innermost
completion technique to the leftmost and outermost cases.

\section{Preliminaries}\label{sect-notations}
 We use the same basic definitions and notions as in~\cite{BaaderN-book98}
 and~\cite{Terese} for TRS and as in~\cite{tata} for tree automata.

\subsection{Terms}
\begin{defi}[Signature]
 A signature is a set whose elements are called function symbols. Each function
 symbol has an arity, which is a natural integer. Function symbols of arity 0
 are called constants. Given a signature $\ensfonc$ and $k\in\N$, the set of its
 function symbols of arity $k$ is denoted by $\ensfonc_k$.
\end{defi}

\begin{defi}[Term, ground term, linearity]
 Given a signature $\ensfonc$ and a set $\ensvar$ whose elements are called variables and such that $\ensfonc\inter\ensvar=\emptyset$, we define the set of terms over $\ensfonc$ and $\ensvar$, $\langtermes\signature$, as the smallest set such that~:
 \begin{enumerate}
  \item $\ensvar\subset\langtermes\signature$ and
  \item $\forall k\in\N, \forall f\in\ensfonc_k, \forall t_1,\dots,t_k\in\langtermes\signature, f(t_1,\dots,t_k)\in\langtermes\signature$.
 \end{enumerate}

 Terms in which no variable appears, \ie terms in
 $\langtermes(\ensfonc,\emptyset)$, are called ground; the set of ground terms
 is denoted $\langtermesclos\ensfonc$. Terms in which any variable appears at
 most once are called linear.\footnote{In particular, any ground term is
   linear.} 
\end{defi}

\begin{defi}[Substitution]
 A substitution over $\langtermes\signature$ is an application from $\ensvar$ to
 $\langtermes\signature$. Any substitution is inductively extended to
 $\langtermes\signature$ by
 $\sigma(f(t_1,\dots,t_k))=f(\sigma(t_1),\dots,\sigma(t_k))$. Given a
 substitution $\sigma$ and a term $t$, we denote $\sigma(t)$ by $t\sigma$.
\end{defi}

\begin{defi}[Context]
 A context over $\langtermes\signature$ is a term in $\langtermes(\ensfonc\cup\ensvar,\{\ctxttriv\})$ in which the variable $\ctxttriv$ appears exactly once. A ground context over $\langtermes\signature$ is a context over $\langtermesclos\ensfonc$. The smallest possible context, $\ctxttriv$, is called the trivial context. Given a context $C$ and a term $t$, we denote $C[t]$ the term $C\sigma_t$, where $\sigma_t:\ctxttriv\mapsto t$.
\end{defi}

\begin{defi}[Position]
 Positions are finite words over the alphabet $\N$. The set of positions of term $t$, $\pos(t)$, is defined by induction over $t$:
 \begin{enumerate}
  \item for all constants $c$ and all variables $X$, $\pos(c)=\pos(X)=\{\Lambda\}$ and
  \item $\pos(f(t_1,\dots,t_k))=\{\Lambda\}\cup\bigcup_{i=1}^k \{i\}.\pos(t_i)$.
 \end{enumerate}
\end{defi}

\begin{defi}[Subterm-at-position, replacement-at-position]
 The position of the hole in context $C$, $\pos_\ctxttriv(C)$, is defined by induction on $C$:
 \begin{enumerate}
  \item $\pos_\ctxttriv(\ctxttriv)=\Lambda$
  \item $\pos_\ctxttriv(f(C_1,\dots,C_k))=i.\pos_\ctxttriv(C_i)$, where $i$ is the unique integer in $\Iff1k$ such that $C_i$ is a context.
 \end{enumerate}

 Given a term $u$ and $p\in\pos(u)$, there is a unique context $C$ and a unique
 term $v$ such that $\pos_\ctxttriv(C)=p$ and $u=C[v]$. The term $v$ is denoted
 by $\ssterme up$, and, given another term $t$, we denote $\rempterme upt=C[t]$.
\end{defi}

\subsection{Rewriting}
\begin{defi}[Rewriting rule, term rewriting system]
 A rewriting rule over $\signature$ is a couple
 $(\ell,r)\in\langtermes\signature\times\langtermes\signature$, denoted by $\ell\to r$, such that any variable appearing in $r$ also appears in $\ell$. A term rewriting system (TRS) over $\signature$ is a set of rewriting rules over $\signature$.
\end{defi}

\begin{defi}[Rewriting step, redex, reducible term, normal form, reflexive and
  transitive closure]\label{defn:reecriture}
 Given a signature $\signature$, a TRS $R$ over it and two terms $s,t\in\langtermesclos\ensfonc$, we say that $s$ can be rewritten into $t$ by $R$, and we note $s\to_R t$ if there exist a rule $\ell\to r\in R$, a ground context $C$ over $\langtermesclos\ensfonc$ and a substitution $\sigma$ over $\langtermes\signature$ such that $s=C[\ell\sigma]$ and $t=C[r\sigma]$.
 
 In this situation, the term $s$ is said to be reducible by $R$ and the subterm
 $\ell\sigma$ is called a redex of $s$. A term $s$ that is irreducible by $R$ is
 a $R$-normal form. The set of terms irreducible by $R$ is denoted $\lirr(R)$.
 We denote $\to_R^*$ the reflexive and transitive closure of $\to_R$.
\end{defi}

\begin{defi}[Set of reachable terms, normalized terms]\label{defn:ensaccessibles}
 Given a signature $\signature$, a TRS $R$ over it and a set of terms
 $L\subset\langtermesclos\ensfonc$, we denote
 $R(L)=\enstq{t\in\langtermesclos\ensfonc}{\exists s\in L, s\to_R t}$, 
 the set of reachable terms $R^*(L)=\enstq{t\in\langtermesclos\ensfonc}{\exists s\in L, s\to_R^* t}$, and
 the set of normalized terms $R^!(L)= R^*(L) \cap \lirr(R)$.
\end{defi}

\begin{defi}[Left-linearity]
 A TRS $R$ is said to be left-linear if for each rule $\ell\to r$ of $R$, the term $\ell$ is linear.
\end{defi}


\subsection{Equations}

\begin{defi}[Equivalence relation, congruence]
 A binary relation is an equivalence relation if it is reflexive, symmetric and transitive.
 An equivalence relation $\equiv$ over $\langtermesclos\ensfonc$ is a congruence
 if for all $k\in\N$, for all $f\in\ensfonc_k$, for all
 $t_1,\dots,t_k,s_1,\dots,s_k\in\langtermesclos\ensfonc$ such that $\forall
 i=1\ldots k$, $t_i\equiv s_i$, we have $f(t_1,\dots,t_k)\equiv f(s_1,\dots,
 s_k)$.
\end{defi}

\begin{defi}[Equation, $\equiv_E$]
  An equation over $\signature$ is a pair of terms
  $(s,t)\in\langtermes\signature\times\langtermes\signature$, denoted by
  $s=t$.  A set $E$ of equations over $\signature$ induces a congruence
  $\equiv_E$ over $\langtermesclos\ensfonc$ which is the smallest congruence
  over $\langtermesclos\ensfonc$ such that for all $s=t\in E$ and for all
  substitutions $\theta:\ensvar\to\langtermesclos\ensfonc$, $s\theta\equiv_E
  t\theta$. The equivalence classes of $\equiv_E$ are denoted with
  $[\cdot]_E$.
\end{defi}

\begin{defi}[Rewriting modulo $E$]\label{defn:ensacc-equation}
  Given a TRS $R$ and a set of equations $E$ both over $\signature$, we define
  the $R$ modulo $E$ rewriting relation, $\to_{R/E}$, as follows. For any
  $u,v\in\langtermesclos\ensfonc$, $u\to_{R/E} v$ if and only if there exist
  $u',v'\in\langtermesclos\ensfonc$ such that $u\equiv_E u'$, $u'\to_R v'$ and
  $v'\equiv_E v$. We define $\to_{R/E}^*$ as the reflexive and transitive
  closure of $\to_{R/E}$, and $(R/E)(L)$ and $(R/E)^*(L)$ in the same way as $R(L)$ and
  $R^*(L)$ where $\to_{R/E}$ replaces $\to_R$.
\end{defi}

\subsection{Tree automata}

 \begin{defi}[Tree automaton, delta-transition, epsilon-transition]
   An automaton over $\ensfonc$ is some $\automate=(\ensfonc,Q,Q_F,\Delta)$
   where $Q$ is a finite set of states (symbols of arity $0$ such that
   $\ensfonc\cap Q=\emptyset$), $Q_F$ is a subset of $Q$ whose elements
   are called final states and $\Delta$ a finite set of transitions. A
   delta-transition is of the form $f(q_1,\dots,q_k)\move q'$ where
   $f\in\ensfonc_k$ and $q_1,\dots,q_k,q'\in Q$. An epsilon-transition is of the
   form $q\move q'$ where $q,q'\in Q$. A configuration of $\automate$ is a term
   in $\langtermes(\ensfonc,Q)$.
 A configuration is elementary if each of its sub-configurations at depth 1 (if any) is a state.
\end{defi}

\noindent
If $\automate=(\ensfonc,Q,Q_F,\Delta)$, by notation abuse, we sometimes write
$q\in\A$ (resp. $s\rw q \in \A$) as a short-hand for $q\in\Q$ (resp. $s \rw q \in \Delta$). We also write $\A
\cup \{s \rw q\}$ for the automaton obtained from $\A$ by adding $q$ to $\Q$ and
$s \rw q$ to $\Delta$.

\begin{defi}
 Let $\automate=(\ensfonc,Q,Q_F,\Delta)$ be an automaton and let $c,c'$ be
 configurations of $\automate$. We say that $\automate$ recognises $c$ into $c'$
 in one step, and denoted by $c\moveun\automate c'$ if there is a transition
 $\tau\move\rho$ in $\automate$ and a context $C$ over $\langtermes(\ensfonc,Q)$
 such that $c=C[\tau]$ and $c'=C[\rho]$. We denote by $\moveet\automate$ the
 reflexive and transitive closure of $\moveun\automate$ and, for any $q\in Q$,
 $\langrecs\automate q=\enstq{t\in\langtermesclos\ensfonc}{t\moveet\automate
   q}$. We extend this definition to subsets of $Q$ and denote it by $\langrec\automate=\langrecs\automate{Q_F}$.
 A sequence of configurations $c_1,\dots,c_n$ such that $t\moveun\A
 c_1\moveun\A\dots\moveun\A c_n\moveun\A q$ is called a recognition path for $t$
 (into $q$) in $\A$. 
\end{defi}

 \begin{exa}\label{ex:completionInnermost}\label{ex:completionInnermost-1}
Let $\ensfonc$ be defined with $\ensfonc_0=\{n,0\}$,  $\ensfonc_1=\{s,a,f\}$,
$\ensfonc_2=\{c\}$ where $0$ is meant to represent integer zero, $s$ the
successor operation on integers, $a$ the predecessor (\enquote{antecessor})
operation, $n$ the empty list, $c$ the constructor of lists of integers and $f$
is intended to be the function on lists that filters out integer zero. 
Let \begin{align*}
R= & \{f(n) \to n, f(c(s(X),Y)) \to c(s(X), f(Y)), f(c(a(X),Y)) \to c(a(X),
f(Y)), \\ & f(c(0,Y)) \to f(Y), a(s(X)) \to X, s(a(X)) \to X \}. \end{align*}
Let $\automate_0$ be the tree automaton with final state
  $q_f$ and transitions $\{n\move q_n, 0\move q_0, s(q_0)\move q_s, a(q_s)\move q_a, c(q_a,q_n)\move q_c, f(q_c)\move q_f\}$. We have $\langrecs{\automate_0}{q_f}=\{f(c(a(s(0)),n))\}$ and $R(\langrecs{\automate_0}{q_f})=\{f(c(0,n)), c(a(s(0)), f(n))\}$.
 \end{exa}
 
\begin{rem}
  In tree automata, epsilon transitions may have \enquote{colors}, like $\colour R$ for
  transition $q\movecol R q'$.  We will use colors $\colour R$ and $\colour E$
  for transitions denoting either rewrite or equational steps.
\end{rem}

\begin{defi}
 Given an automaton $\A$ and a color $\colour R$, we denote by $\asc\automate R$ the automaton obtained from $\automate$ by removing all transitions colored with $\colour R$.
\end{defi}

\begin{defi}[Determinism, Completeness, Accessibility]
  An automaton is deterministic if it has no epsilon-transition and for all
  delta-transitions $\tau\move\rho$ and $\tau'\move\rho'$, if $\tau=\tau'$ then
  $\rho=\rho'$. An automaton is complete if each of its elementary
  configurations is the left-hand side of some of its transitions. A state $q$
  of automaton $\A$ is accessible if $\langrecs\A q\neq\emptyset$. An automaton
  is accessible if all of its states are.
\end{defi}

\begin{defi}[Equivalence relation on states and configurations]
 Given two states $q$, $q'$ of some automaton $\automate$ and a color $\colour
 E$, we note $q\moverepcol \automate E q'$ when we have both $q\moveuncol
 \automate Eq'$ and $q'\moveuncol \automate E q$. This relation is extended to a congruence relation over $\langtermes(\ensfonc,Q)$. The equivalence classes are noted with $\classcol E\cdot$.
\end{defi}

\begin{exa} 
Let $\automate$ be the tree automaton with transitions $a \move q_0$, $b \move
q_1$, $s(q_0) \move q_2$, $q_0 \moveuncol{\automate}{E}{q_1}$ and $q_1
\moveuncol{\automate}{E}{q_0}$. The equivalence class $\classcol E {q_0}$ contains
$q_0$ and $q_1$. The equivalence class $\classcol E {s(q_0)}$ contains
configurations $s(q_0)$ and $s(q_1)$.
\end{exa}
\begin{rem}
 $q\moverepetcol\automate Eq'$ is stronger than $(q\moveetcol\automate E q'
 \land q'\moveetcol\automate E q)$. $\moverepetcol\automate E$ is an equivalence
 relation over $Q_\automate$. 
\end{rem}

\begin{defi}\label{defn-quotrep}
 Let $\automate=(\ensfonc, Q,Q_F,\Delta)$ be an automaton and $\colour E$ a color. We note $\quotcol\automate E$ the automaton over $\ensfonc$ whose set of states is $\quotcol QE$, whose set of final states is $\quotcol{Q_F} E$ and whose set of transitions is {\small\[\enstq{f(\classcol E{q_1},\dots,\classcol E{q_k})\move\classcol E{q'}}{f(q_1,\dots,q_k)\move q'\in\Delta} \cup\enstq{\classcol E q\move\classcol E{q'}}{q\move q'\in\Delta\land \classcol Eq\neq\classcol E{q'}}.\]}
\end{defi}
\begin{rem}
 For any configurations $c,c'$ of $\automate$, we have $c\moveet\automate c'$ if and only if $\classcol Ec\moveet{\quotcol\automate E}\classcol E{c'}$. So the languages recognised by $\automate$ and $\quotcol\automate E$ are the same.
\end{rem}

\subsection{Pair automaton}

We now give notations used for pair automaton, the archetype of which is the product of two automata.

\begin{defi}[Pair automaton]
 An automaton $\automate=(\ensfonc,Q,Q_F,\Delta)$ is said to be a pair
 automaton if there exists some sets $Q_1$ and $Q_2$ such that $Q=Q_1\times Q_2$.
\end{defi}

\begin{defi}[Product automaton~\cite{tata}]
 Let $\automate=(\ensfonc,Q,Q_F,\Delta_\automate)$ and
 $\bautomate=(\ensfonc,P,P_F,\Delta_\bautomate)$ be two automata. The product
 automaton of $\automate$ and $\bautomate$ is
 $\automate\times\bautomate=(\ensfonc,Q\times P,Q_F\times P_F,\Delta)$ where 
{\small\[
\begin{array}{r@{}l}
  \Delta= & \{f(\paire(q_1,p_1),\dots,\paire(q_k,p_k))\move\paire(q',p')
\sep f(q_1,\dots,q_k)\move q'\in\Delta_\automate \: \land \:
f(p_1,\dots,p_k)\move p'\in\Delta_\bautomate\}\: \cup  \\ 
& \enstq{\paire(q,p)\move\paire(q',p)}{p\in P, \: q\move q'\in\Delta_\automate} \cup\enstq{\paire(q,p)\move\paire(q,p')}{q\in Q, \: p\move p'\in\Delta_\bautomate}.
\end{array}
\]}
\end{defi}


\begin{defi}[Projections]
 Let $\automate=(\ensfonc,Q,Q_F,\Delta)$ be a pair automaton, let $\tau\move\rho$ be one of its transitions and $\paire(q,p)$ be one of its states. We define $\projgauche{\paire(q,p)}=q$ and extend $\projgauche\cdot$ to configurations inductively: $\projgauche{f(\gamma_1,\dots,\gamma_k)}=f(\projgauche{\gamma_1},\dots,\projgauche{\gamma_k})$. We define $\projgauche{\tau\move\rho}=\projgauche\tau \move \projgauche\rho$.  We define $\projgauche\automate=(\ensfonc,\projgauche Q,\projgauche{Q_F},\projgauche\Delta)$. $\projdroite\cdot$ is defined on all these objects in the same way for the right component.
\end{defi}

\begin{rem}\label{rem-proj}
 Using $\projgauche\automate$ amounts to forgetting the precision given by the
 right component of the states. As a result,
 $\langrecs{\projgauche\automate}q\supset \bigcup_{p\in P}
 \langrecs\automate{\paire(q,p)}$.
\end{rem}

\subsection{Innermost strategy}
In general, a strategy over a TRS $R$ is a set of (computable) criteria to
describe a certain sub-relation of $\to_R$. In this paper, we will be interested
in innermost strategies. In these strategies, commonly used to execute
functional programs (\enquote{call-by-value}), terms are rewritten by always
contracting one of the lowest reducible subterms. If $s
\to_R t$ and rewriting occurs at a position $p$ of $s$, recall that $s|_p$ is called the {\em redex}.

\begin{defi}[Innermost strategy]
 Given a TRS $R$ and two terms $s, t$, we say that $s$ can be rewritten into $t$
 by $R$ with an innermost strategy, denoted by $s\to_{\Rin}t$, if $s\to_R t$
 and each strict subterm of the redex in $s$ is a $R$-normal form. We define
 $\Rin(L)$ and $\Rin^*(L)$ in the same way as $R(L)$, $R^*(L)$ where
 $\to_{\Rin}$ replaces $\to_R$.  
\end{defi}

\begin{exa}\label{ex:completionInnermost-2}
 We continue on Example~\ref{ex:completionInnermost}. We have
 $\Rin(\langrecs{\automate_0}{q_f})=\{f(c(0,n))\}$ because the rewriting step
 $f(c(a(s(0)), n)) \to_R c(a(s(0)), f (n))$ is not innermost since the subterm
 $a(s(0))$ of the redex $f(c(a(s(0)), n))$ is not in normal form.
\end{exa}

%

\noindent
To deal with innermost strategies, we have to discriminate normal forms. 
When $R$ is left-linear, it is possible to compute a tree automaton recognising
normal forms~\cite{RemyComon87}. This automaton can be computed in an efficient
way using~\cite{comon-IC00}.
\begin{thm}[\cite{RemyComon87}]\label{thm-airr}
 Let $R$ be a left-linear TRS. There is a deterministic and complete tree
 automaton $\airr(R)$ whose states are all final except one, denoted by $\pred$ and such that $\langrec{\airr(R)}=\lirr(R)$ and $\langrecs{\airr(R)}\pred=\langtermesclos\ensfonc\prive\lirr(R)$.
\end{thm}
\begin{rem}
 Since $\airr(R)$ is deterministic, for any state $p\neq\pred$,
 $\langrecs{\airr(R)}p\subset\lirr(R)$. 
\end{rem}

\begin{rem}\label{rem-omit-pred}
  If a term $s$ is reducible, any term having $s$ as a subterm is also
  reducible. Thus any transition of $\airr(R)$ where $\pred$ appears in the
  left-hand side will necessarily have $\pred$ as its right-hand
  side. Thus, for brevity, these transitions will always be left implicit
  when describing the automaton $\airr(R)$ for some TRS $R$.
\end{rem}

\begin{exa}\label{ex:completionInnermost-3}
  In Example~\ref{ex:completionInnermost}, $\airr(R)$ needs, in addition to
  $\pred$, a state $p_{list}$ to recognise lists of integers, a state $p_a$ for
  terms of the form $a(\dots)$, a state $p_s$ for $s(\dots)$, a state $p_0$ for
  $0$ and a state $p_{var}$ to recognise terms that are not subterms of
  left-hand sides of $R$, but may participate in building a reducible term by
  being instances of variables in a left-hand side. We note $P=\{p_{list}, p_0,
  p_a, p_s, p_{var}\}$ and $P_{int}=\{p_0,p_a,p_s\}$. The interesting
  transitions are thus \[\begin{array}{lllll}
    0\move p_0 & \hspace*{1cm}&\textstyle\bigcup_{p\in P\prive\{p_a\}}\{s(p)\move p_s\} & \hspace*{1cm}&
    \textstyle\bigcup_{p\in P\prive\{p_s\}}\{a(p)\move p_a\} \\
    n\move p_{list} & & \textstyle\bigcup_{p\in P_{int}, p'\in P}\{c(p, p')\move
    p_{list}\} & & f(p_{list})\move\pred \\ a(p_s)\move\pred & &
    s(p_a)\move\pred. & &
  \end{array}\] Furthermore,
  as remarked above, any configuration that contains $\pred$ is recognised
  into $\pred$. Finally, some configurations are not covered by the previous
  cases: they are recognised into $p_{var}$.
\end{exa}

\section{Innermost equational completion}\label{sect-innermost}
Our first contribution is an adaptation of the classical equational completion of~\cite{GenetR-JSC10}, which is
an iterative process on automata. Starting from a tree automaton $\A_0$ it
iteratively computes tree automata $\A_1,\A_2,\ldots$ until a fixpoint automaton
$\A_*$ is found. Each
iteration comprises two parts: (exact) completion itself (Subsection~\ref{ssect-exact}), then equational
merging (Subsection~\ref{ssect-equations}). The former tends to incorporate descendants by $R$ of already
recognised terms into the recognised language; this leads to the creation of new
states. The latter tends to merge states in order to ease termination of the
overall process, at the cost of precision of the computed result. In the
completion procedure proposed here, some
transition added by equational completion will have colors $\colour R$ or
$\colour E$. We will use colors $\colour R$ and $\colour E$
  for transitions denoting either rewrite or equational steps; it is assumed
  that the transitions of the input automaton $\A_0$ do 
not have any color and that $\A_0$ does not have any epsilon-transition.

The equational completion of~\cite{GenetR-JSC10} is blind to strategies. To make
it innermost-strategy-aware, we equip each state of the studied automaton with a
state from the automaton $\airr(R)$ (see Theorem~\ref{thm-airr}) to keep track
of normal and reducible forms. 
Let $\autoInit$ be an automaton recognising the initial
language. Completion will start with $\automate_0 = \autoInit \times
\airr(R)$. Since the $\airr(R)$ component of this product automaton is
complete, the product enjoys the following property.

\begin{lem}\label{lem-proj-eq}
 If $\automate$ and $\bautomate$ are two tree automata, and $\bautomate$ is
 complete, then 
\[\langrecs{\projgauche{\automate\times\bautomate}}q = \bigcup_{p\in P}
 \langrecs{\automate\times\bautomate}{\paire(q,p)}.\]
\end{lem}
\begin{proof}
Proving the inclusion of the right-hand side in the left-hand side uses
Remark~\ref{rem-proj}. For the other direction, let $t$ be a term belonging to 
$\langrecs{\projgauche{\automate\times\bautomate}}q$. We know that $t
\moveet\A q$. Besides, since $\bautomate$ is complete, we know that there exists
a state $p$ of $\bautomate$ such that $t \moveet\B p$. 
\end{proof}

\noindent
In the following, automata built by completion will enjoy consistency with $\airr(R)$, we now define.

\begin{defi}[Consistency with $\airr(R)$]
 A pair automaton $\automate$ is said to be consistent with $\airr(R)$ if, for any configuration $c$ and any state $\paire(q,p)$ of $\automate$, $\projdroite c$ is a configuration of $\airr(R)$ and $p$ is a state of $\airr(R)$, and if $c\moveet\A\paire(q,p)$ then $\projdroite c\moveet{\airr(R)} p$.
\end{defi}

\subsection{Exact completion}\label{ssect-exact}

The first step of equational completion incorporates descendants by $R$ of terms
recognised by $\A_i$ into $\A_{i+1}$. The principle is to search for critical
pairs between $\A_i$ and $R$. In classical completion, a critical pair is triple
$(\ell\to r, \sigma, q)$ such that $l\sigma \moveet{\automate_i} q$, $l\sigma
\rwR r\sigma$ and $r\sigma \not\moveet{\automate_i} q$. Such a critical pair
denotes a rewriting position of a term recognised by $\A_i$ such that the rewritten term
is not recognised by $\A_i$. For the innermost strategy, the critical pair
notion is slightly refined since it also needs that every subterm $t$ at
depth $1$ in $\ell\sigma$ is in normal form. This corresponds to the third case
of the following definition where
$t_i\moveet\automate\paire(q_i,p_i)$ and $p_i\neq p_{red}$
ensures that all $t_i$'s are irreducible. See Figure~\ref{fig:completion}.

\begin{defi}[Innermost critical pair]\label{def-cpn-eff}
Let $\A$ be a pair automaton. 
 A tuple $(\ell\to r, \sigma, \paire(q,p))$ where $\ell\to r\in R$,
 $\sigma:\ensvar\to Q_\automate$ and $\paire(q,p)\in Q_\automate$ is called a
 critical pair if
 \begin{enumerate}
  \item $\ell\sigma\moveet\automate\paire(q,p)$,
  \item there is no $p'$ such that $r\sigma\moveet\automate\paire(q,p')$ and
  \item let $f\in\F$. For $1 \leq i \leq n$, let $t_i$ be terms, and
    $\paire(q_i,p_i)$ be states such that $l\sigma = f(t_1,\ldots,t_n)$ and 
    $f(t_1,\ldots,t_n) \moveet\automate f(\paire(q_1,p_1), \ldots,
    \paire(q_n,p_n)) \moveet\automate \paire(q,p)$. For all $1 \leq i \leq n$,
    $p_i\neq \pred$.
 \end{enumerate}

\end{defi}

\begin{rem}
 Because a critical pair denotes a rewriting situation, the $p$ of Definition~\ref{def-cpn-eff} is necessarily $\pred$ as long as $\automate$ is consistent with $\airr(R)$.
\end{rem}

\begin{exa}\label{ex:completionInnermost-4}
In the situation of
Examples~\ref{ex:completionInnermost-1} and~\ref{ex:completionInnermost-3},
consider the rule $f(c(a(X),Y))\to c(a(X), f(Y))$, the substitution
$\sigma_1=\{X\mapsto \paire(q_s,p_s), Y\mapsto \paire(q_n,p_n)\}$ and the state
$\paire(q_f,\pred)$: this is not an innermost critical pair because the
recognition path is:

$f(c(a(\paire(q_s,p_s)),
\paire(q_n,p_n)))\move f(c(\paire(q_a,\pred),\paire(q_n,p_n)))\move f(\paire(q_c,\pred))\move
 \paire(
q_f,\pred)$ 

\noindent
where there is a $\pred$ at depth 1. This is due to the fact that
$a(\paire(q_s,p_s)) \move \paire(q_a,\pred)$ recognizes a term of the form
$a(s(0))$ which is reducible.
But there is an innermost critical pair in $\A_0$ with the rule $a(s(X))\to X$,
the substitution $\sigma_2=\{X\mapsto \paire(q_0,p_0)\}$ and the state
$\paire(q_a,\pred)$. The recognition path is here $a(s(\paire(q_0,p_0))) \move
a(\paire(q_s,p_s)) \move \paire(q_a,\pred)$.\, where at depth $1$ the term
$s(\paire(q_0,p_0))$ is recognized into state $\paire(q_s,p_s)$ and $p_s\neq \pred$.
\end{exa}

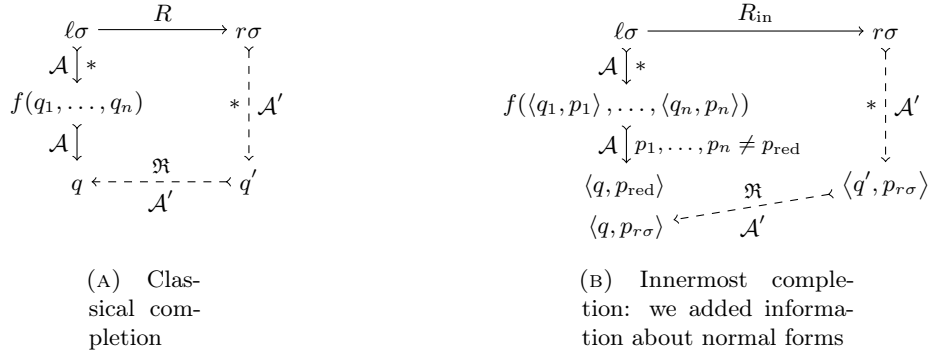
\begin{figure}[!ht]
\hskip 0pt plus .5fill
\setbox4=\hbox{%
\footnotesize
\begin{tikzpicture}
\matrix [column sep = {1cm}, row sep = {.5cm}]{
\node (ls) {$\ell\sigma$}; & \node (rs) {$r\sigma\vphantom{\ell}$}; \\
\node (qd1) {$f(q_1,\dots,q_n)$}; &\\ 
\node (q)  {$q\vphantom{'}$}; & \node (qp) {$q'$};\\[-.5cm]
\node (qprs) {\strut}; &\\
};

\path[->] (ls) edge node [above] {$R$} (rs);
\path[>->] (ls) edge node [left] {$\automate$} node [right] {\scriptsize$*$} (qd1);
\path[>->] (qd1) edge node [left] {$\automate$} (q);
\path[>->, dashed] (rs) edge node [right] {$\automate'$} node [left] {\scriptsize$*$} (qp);
\path[>->, dashed] (qp) edge node [below] {$\automate'$} node [above] {\scriptsize $\colour R$} (q);
\end{tikzpicture}}
\subfigure[Classical completion]{\box4}
\hfill
\setbox2=\hbox{%
\footnotesize
 \begin{tikzpicture}
\matrix [column sep = {1cm}, row sep = {.5cm}]{
\node (ls) {$\ell\sigma$}; & \node (rs) {$r\sigma\vphantom{\ell}$}; \\
\node (qd1) {$f(\paire(q_1,p_1), \dots,\paire(q_n,p_n))$}; &\\ 
\node (q)  {$\paire(q,\pred)\vphantom{'}$}; & \node (qp) {$\paire(q',p_{r\sigma})$};\\[-.5cm]
\node (qprs) {$\paire(q,p_{r\sigma})$}; & \\
};

\path[->] (ls) edge node [above] {$\Rin$} (rs);
\path[>->] (ls) edge node [left] {$\automate$} node [right] {\scriptsize$*$} (qd1);
\path[>->] (qd1) edge node [left] {$\automate$} node [right] {\scriptsize $p_1,\dots,p_n\neq\pred$} (q);
\path[>->, dashed] (rs) edge node [right] {$\automate'$} node [left] {\scriptsize$*$} (qp);
\path[>->, dashed] (qp) edge node [below] {$\automate'$} node [above] {\scriptsize $\colour R$} (qprs);
\end{tikzpicture}}
\subfigure[Innermost completion: we added information about normal forms]{\box2}
\hskip 0pt plus.5fill\null
\caption{Comparison of classical and innermost critical pairs}
\label{fig:completion}
\end{figure}

\noindent
Once a critical pair is found, the completion algorithm needs to resolve it: it
adds the necessary transitions for $r\sigma$ to be recognised by the completed
automaton.  Classical completion adds the necessary transitions so that $r\sigma
\moveet{\A'} q$, where $\A'$ is the completed automaton. In innermost completion
this is more complex. The state $q$ is, in fact, a pair of the form
$\paire(q,p_{red})$ and adding transitions so that $r\sigma \moveet{\A'}
\paire(q,p_{red})$ may jeopardise consistency of $\A'$ with $\airr$ if 
$r\sigma$ is not reducible. Thus the diagram is closed in a different way
preserving consistency with $\airr$ (see Figure~\ref{fig:completion}).  However,
like in classical completion, this can generally not be done in one step, as
$r\sigma$ might be a non-elementary configuration. We have to split the
configuration into elementary configurations and to introduce new states to
recognise them: this is what {\em normalisation} (denoted by
$\normalisation{\A}$) does. Given an automaton $\A$, a configuration $c$ and a
new state $\paire(q,p)$, we denote by 
$\normalisation{\A}(c, \paire(q,p))$ the set of transitions 
that we add to $\A$ to ensure that $c$ is recognised into $\paire(q,p)$. The
$\normalisation{\A}$ operation is parameterized by $\A$ because it
reuses transitions of $\A$ whenever it is possible and adds new transitions
and new states otherwise.

\begin{defi}[New state] Let $\A=(\ensfonc,\Q,\Q_F,\Delta)$ be a
  product automaton.
A new state (for $\A$) is a fresh symbol not occurring in $\ensfonc \cup \Q$.
\end{defi}

\noindent
We here define normalisation as a bottom-up process. In the recursive call, the
choice of the context $C[\,]$ may be non deterministic but all the possible
results are equivalent modulo a state renaming. Recall that the notation $\A \cup
\{f(q_1,\ldots,q_n) \rw q'\}$ is a short-hand denoting the automaton 
obtained from $\A$ by adding $q$ to its set of states $\Q$ and $s \rw q$ to its
set of transitions.

\begin{defi}[Normalisation]
\label{def:normalisation}
Let $\A$ be a product automaton with set of states $Q$, and such that
$\projdroite{\A}$ is the automaton $\airr(R)$ (for a TRS $R$). 
Let $C[\;] \in \TFQ \setminus \Q$ be a non empty context built on states of
$\A$. Let $f\in \ensfonc$ of arity $n$, $q_1, \ldots, q_n$ states of $\A$ and
$q$ a new state for $\A$. The normalisation function is inductively defined by:
\begin{enumerate}
\item $\normalisation{\A}(f(q_1, \ldots, q_n), q) = \{ f(q_1, \ldots, q_n) \rw
  q\}$

\item $\normalisation{\A}(C[f(q_1, \ldots, q_n)], q) =$ 
\begin{tabular}[t]{l}
    $ \{ f(q_1, \ldots, q_n) \rw q'\} \: \cup$ 
    $\normalisation{\A\cup\{f(q_1, \ldots, q_n)\rw q'\}} (C[q'], q)$
\end{tabular} 
\end{enumerate}

\medskip
\noindent
where $\left\{\begin{array}{l}q'\in\A \\ f(q_1, \ldots, q_n) \rw q' \in \A\end{array}\right.$ or 
$\left\{\begin{array}{l}
\mbox{$q'$ is a new state for $\A$, and} \\
\mbox{$f(\projdroite{q_1},\ldots,\projdroite{q_n}) \rw \projdroite{q'} \in \airr(R)$, and } \\
\mbox{$\forall q'' \in Q : f(q_1, \ldots, q_n) \rw q'' \not \in \A$.}
\end{array}\right.$
\end{defi}

\noindent
In the above definition, for any new state $q'=\paire(r,p)$ used to normalize a
subterm $f(q_1, \ldots, q_n)$, $r$ is the only part of the state that is
arbitrary. Indeed, $p$ is fixed by $\airr(R)$. The term
$f(\projdroite{q_1},\ldots, \projdroite{q_n})$ is necessarily the left-hand side
of a transition of $\airr(R)$ (it is complete) and $p$ is
the right-hand side of this transition. This is necessary for normalisation to
preserve consistency with $\airr(R)$.
\begin{exa}
  With a suitable signature, suppose that automaton $\A$ consists of the
  transitions $c\move\paire(q_1,p_c)$ and $f(\paire(q_1,p_c))\move\paire(q_2,
  p_{f(c)})$ and we want to normalise $f(g(\paire(q_2,p_{f(c)}), c))$ to the new
  state $\paire(q_N, p_{f(g(f(c), c))})$. We first have to normalise under $g$:
  $\paire(q_2,p_{f(c)})$ is already a state, so it does not need to be
  normalised; $c$ has to be normalised to a state: since $\A$ already has
  transition $c\move\paire(q_1,p_c)$, we add no new state and it remains to
  normalise $g(\paire(q_2,p_{f(c)}), \paire(q_1,p_c))$. Since $\A$ does not
  contain a transition for this configuration, we must add a new state
  $\paire(q', p_{g(f(c), c)})$ and the transition $g(\paire(q_2,p_{f(c)}),
  \paire(q_1,p_c))\move\paire(q', p_{g(f(c), c)})$. Finally, we add
  $f(\paire(q', p_{g(f(c), c)}))\move\paire(q_N, p_{f(g(f(c), c))})$. Note that
  due to consistency with $\airr(R)$, whenever we add a new transition
  $c'\move\paire(q',p')$, only the $q'$ is arbitrary: the $p'$ is always the
  state of $\airr(R)$ such that $\projdroite c\moveun{\airr(R)}p'$, in order to
  preserve consistency with $\airr(R)$.
\end{exa}

%

\noindent
Completion of a critical pair is done in two steps. The first set of operations
formalises \enquote{closing the square} (see  Figure~\ref{fig:completion}), \ie
if $l\sigma \moveet\A \paire(q,p_{red})$ then we add transitions $r\sigma \moveet{\A'}
\paire(q',p_{r\sigma}) \movecol R\paire(q,p_{r\sigma})$. The second step adds
the necessary transitions for any context $C[r\sigma]$ to be recognised in the
tree automaton if $C[l \sigma]$ was. Thus if the recognition path for
$C[l\sigma]$ is of the form $C[l\sigma] \moveet\A C[\paire(q,p_{red})] \moveet\A
\paire(q_c,p_{red})$, we add the necessary transitions for
$C[\paire(q,p_{r\sigma})]$ to be recognised into $\paire(q_c,p_c)$ where $p_c$
is the state of $\airr(\R)$ recognising $C[r\sigma]$.

\begin{defi}[Completion of an innermost critical pair]\label{defn-completion}
 A critical pair $(\ell\to r, \sigma, \paire(q,p))$ in automaton $\automate$ is
 completed by first computing $N=\normalisation{\automate}(r\sigma,
 \paire(q',p_{r\sigma}))$ where $q'$ is a new state and $p_{r\sigma}$ is the
 state such that 
 $\projdroite{r\sigma}\moveet{\airr(R)} p_{r\sigma}$, then adding to $\automate$
 the new states and the transitions appearing in $N$ as well as the transition
 $\paire(q',p_{r\sigma})\movecol R\paire(q,p_{r\sigma})$. If $r\sigma$ is a
 trivial configuration (\ie $r$ is just a variable, and thus
 $\projdroite{r\sigma}$  is a state), only transition $r\sigma\movecol
 R\paire(q,\projdroite{r\sigma})$ is added. 
 Afterwards, we execute the following supplementary operations. For any new
 transition $f(\dots,\paire(q,\pred),\dots)\move \paire(q'',p'')$, we add a
 transition $f(\dots,\paire(q,p_{r\sigma}),\dots)\move\paire(q'',p''')$ with
 $f(\dots,p_{r\sigma},\dots)\moveun{\airr(R)}p'''$. These new transitions are in
 turn recursively considered for the supplementary operations\footnote{Those
   supplementary operations add new pairs, but the element of each pair are not
   new. So, this necessarily terminates.}.
\end{defi} 

\begin{defi}[Innermost completion step]
 Let $PC$ be the set of all innermost critical pairs of $\A_i$. For $pc\in PC$,
 let $N_{pc}$ be the set of new states and transitions needed under
 Definition~\ref{defn-completion} to complete $pc$, and $\A \cup
 N_{pc}$ the automaton $\A$ completed by states and transitions of $N_{pc}$. Then $\A_{i+1}=\A_{i}\union\bigcup_{pc\in PC}N_{pc}$.
\end{defi}

\begin{lem}\label{lemme-projdroite-irr}
 Let $\automate$ be an automaton obtained from some $\autoInit\times{\airr(R)}$ after some steps of innermost completion. $\automate$ is consistent with ${\airr(R)}$.
\end{lem}

\begin{proof}
  $\autoInit\times{\airr(R)}$ is consistent with ${\airr(R)}$ by construction.
  Adding a new transitions, during completion, also preserves the consistency because we choose $p'$ such that
  $r\sigma\moveet{\airr(R)} p'$. This is also the case for transitions built by
  normalisation because for any intermediate subterm $t$ normalized into
  $\paire(q,p)$, $p$ is chosen such that $\projdroite{t}
  \moveet{\airr(R)} p$. The same goes for the supplementary operations.
\end{proof}

\subsection{Equational simplification}\label{ssect-equations}

\begin{defi}[Situation of application of an equation]\label{defn-situation-application}
 Given an equation $s=t$, an automaton $\automate$, a substitution $\theta:\ensvar\to Q_\automate$ and states $\paire(q_1,p_1)$ and $\paire(q_2,p_2)$, we say that $(s=t,\theta, \paire(q_1,p_1), \paire(q_2,p_2))$ is a situation of application in $\automate$ if 
 \begin{enumerate-ligne}
  \item $s\theta\moveet\automate \paire(q_1,p_1)$,
  \item $t\theta\moveet\automate\paire(q_2,p_2)$,
  \item $\paire(q_1,p_1)\notmoverepcol\automate E\paire(q_2,p_2)$
  \anditem $p_1=p_2$.
 \end{enumerate-ligne}
\end{defi}
\noindent
Note that when $p_1\neq p_2$, this is not a situation of application for an
equation. This restriction avoids, in particular, to apply an 
equation between reducible and irreducible terms, and thus preserves consistency
of the completed automaton w.r.t. $\airr(R)$. Such terms will be recognised
by states having two distinct second components. On the opposite, when a
situation of application arises, we \enquote{apply} the 
equation, \ie add the necessary transitions to have
$\paire(q_1,p_1)\moverepcol\automate E\paire(q_2,p_2)$ and supplementary
transitions to lift this property to any embedding context.  We apply equations
until there are no more situation of application on the automaton (this is
guaranteed to happen because we add no new state in this part).
\begin{defi}[Application of an equation]\label{defn-situation-equation}
 Given $(s=t,\theta, \paire(q_1,p_1), \paire(q_2,p_1))$ a situation of application in $\automate$, applying the underlying equation in it consists in adding transitions $\paire(q_1,p_1)\movecol E\paire(q_2,p_1)$ and $\paire(q_2,p_1)\movecol E\paire(q_1,p_1)$ to $\automate$. We also add the supplementary transitions $\paire(q_1,p'_1)\movecol E\paire(q_2,p'_1)$ and $\paire(q_2,p'_1)\movecol E\paire(q_1,p'_1)$ where $\paire(q_1,p'_1)$ and $\paire(q_2,p'_1)$ occur in the automaton.
\end{defi}


\begin{lem}\label{lemme-eq-cons-irr}
 Applying an equation preserves consistency with $\airr(R)$.
\end{lem}
\begin{proof}
  Let $\automate$ be a consistent with $\airr(R)$ automaton whose set of states
  is $Q$, let $\bautomate$ be the automaton resulting from the addition of
  transition $\paire(q_1,p_1)\move\paire(q_2,p_1)$ to $\automate$ due to the
  application of some equation. Note that this is sufficient because of the
  symmetry between $q_1$ and $q_2$. We proceed by induction on $k$, the number
  of times the transition $\paire(q_1,p_1)\move\paire(q_2,p_1)$ occurs in the
  path $c\moveet\bautomate\paire(q,p)$ where $c$ is a configuration and
  $\paire(q,p)$ is a state of $\bautomate$. If there is no occurrence, then
  $c\moveet\automate\paire(q,p)$ and by consistency of $\automate$, $\projdroite
  c\moveet{\airr(R)}p$.
 
  Suppose the property is true for some $k$ and there is a context $C$ on
  $\langtermes(\ensfonc,Q)$ such that $c\moveet\automate
  C[\paire(q_1,p_1)]\moveun\bautomate
  C[\paire(q_2,p_1)]\moveet\bautomate\paire(q,p)$ with the last part of the path
  using less than $k$ times the new transition. First, there is a configuration
  $c_1$ such that $c=C[c_1]$ and $c_1\moveet\automate\paire(q_1,p_1)$, and
  therefore $\projdroite{c_1}\moveet{\airr(R)}p_1$ by consistency of $\automate$
  w.r.t. $\airr(R)$. Second, by induction
  hypothesis, $\projdroite{C[p_1]}\moveet{\airr(R)} p$. Finally, $\projdroite
  c\moveet{\airr(R)} \projdroite{C[p_1]}\moveet{\airr(R)}p$.
\end{proof}

\subsection{Innermost completion and equations}

\begin{defi}[Step of innermost equational completion]\label{completion-algo}
Let $R$ be a left-linear TRS, $\autoInit$ a tree automaton, $E$ a set of equations
and $\automate_0 =\autoInit\times{\airr(R)}$. The automaton $\A_{i+1}$ is
obtained, from $\A_i$, by applying an innermost completion step on $\A_i$
(Definition~\ref{defn-completion}) and solving all situations of applications of
equations of $E$ (Definition~\ref{defn-situation-application}). 
\end{defi}

\section{Correctness}\label{sect-th}

The objective of this part is to prove that, when completion terminates,
the produced tree automaton is closed w.r.t. innermost rewriting, \ie if the
completed automaton $\automate$ recognizes a term $s$ and $s \to_{Rin} t$ then
$\automate$ also recognizes $t$. To prove this property we rely on the notion of
{\em correct automaton}. An automaton is correct if it is closed by
innermost rewriting or if it still contains critical pairs to solve.

\begin{defi}[Correct automaton]\label{defn-correct}
 An automaton $\automate$ is correct w.r.t. $\Rin$ if for all states $\paire(q,\pred)$ of $\automate$, for all $u\in\langrecs\automate{\paire(q,\pred)}$ and for all $v\in \Rin(u)$, either 
 there is a state $p$ of $\airr(R)$ such that $v\in\langrecs\automate{\paire(q,p)}$
 or there is a critical pair $(\ell\to r,\sigma,\paire(q_0,p_0))$ in $\automate$ for some $\paire(q_0,p_0)$ and a context $C$ on $\langtermesclos\ensfonc$ such that $u\moveet\automate C[\ell\sigma]\moveet\automate C[\paire(q_0,\pred)]\moveet\automate \paire(q,\pred)$ and $v\moveet\automate C[r\sigma]$.
\end{defi}

\noindent
First, we show that correction of automata is preserved by equational
simplification. 

\begin{lem}[Simplification preserves correction]
\label{lem-simp-correct}
Let $\A$ be an automaton correct w.r.t. $\Rin$. If $\A'$ is the result of the 
equational simplification of $\A$, then $\A'$ is correct.
\end{lem}
\begin{proof}
By definition of equational simplification, we trivially have $\moveet\A
\subseteq \moveet{\A'}$. Thus, for all states $\paire(q,\pred)$ of $\A$ if
$u\in\langrecs\A{\paire(q,\pred)}$, $v\in\Rin(u)$ and $v \in
\langrecs\A{\paire(q,p)}$ we have $u \in \langrecs{\A'}{\paire(q,\pred)}$ and
$v \in \langrecs{\A'}{\paire(q,p)}$. Similarly, for the other case of
correction, if $u\moveet\automate
C[\ell\sigma]\moveet\automate C[\paire(q_0,\pred)]\moveet\automate
\paire(q,\pred)$ and $v\moveet\automate C[r\sigma]$, we have the same
derivations in $\A'$
\end{proof}

\noindent
Now we can show that any automaton produced by completion is correct. This is 
Lemma~\ref{lemme-correct} proven below. However, to achieve this proof we need the
following intermediate lemma to ensure that if $C[\ell\sigma]$ rewrites to 
$C[r\sigma]$, using a rewrite rule $ \ell\to r$, and if $C[\ell\sigma]
\moveet\automate\paire(q_1,p_1)$, then there will be a state $p_2$ in the
completed automaton $\bautomate$ such that
$C[\paire(q,p_{r\sigma})]\moveet\bautomate\paire(q_1,p_2)$.

\begin{lem}\label{lemme-composition-contexte}
 Let $\automate$ be an automaton consistent with $\airr(R)$, $(\ell\to r, \sigma, \paire(q,p))$ a critical pair in $\automate$, let $p_{r\sigma}$ be the state of $\airr(R)$ such that $\projdroite{r\sigma}\moveet{\airr(R)}p_{r\sigma}$ and $\bautomate$ be the automaton resulting from the completion of this critical pair. Let $C$ be a context on $\langtermesclos\ensfonc$ and $\paire(q_1,p_1)$ a state of $\automate$ such that $C[\paire(q,p)]\moveet\automate\paire(q_1,p_1)$. Then there exists a state $p_2$ of $\airr(R)$ such that $C[\paire(q,p_{r\sigma})]\moveet\bautomate\paire(q_1,p_2)$.
\end{lem}
\begin{proof}
  Note that we necessarily have $p=p_1=\pred$ since $l\sigma$ is reducible and
  so is $C[l\sigma]$.  We have to show that all the transitions used in the path
  $C[\paire(q,\pred)]\moveet\automate\paire(q_1,\pred)$ have some counterpart
  starting from $C[\paire(q,p_{r\sigma})]$. First, observe that all transitions
  used to recognise subterms at positions of $C$, that are parallel to the
  position of the hole, remain unchanged. Second, if some transition only
  comprises states whose left component are in $\autoInit$, then it has a
  counterpart for any choice of right components in $\airr(R)$ because our
  automaton contains the whole product $\autoInit\times\airr(R)$. It remains to
  show that the transitions involving new states added by the normalisation
  during the completion of the considered critical pair also have their
  counterpart: they exist thanks to the supplementary operations of
  Definition~\ref{defn-completion}.
\end{proof}

\begin{rem}
 Note that the supplementary operations described in the completion algorithm
 are necessary for this lemma to hold. Indeed, take $R=\{g(f(b))\to g(f(a)), f(a)\to
 c\}$, $\autoInit=\{b\move q_b, f(q_b)\move q_{fb}, g(q_{fb})\move
 q_{gfb}\}$. We have 
\begin{align*}
\airr(R)=& \{a\move p_a, b\move p_b, c\move p_c,
 f(p_a)\move\pred, g(p_a)\move p_c, \\ & f(p_b)\move p_{fb}, g(p_{fb})\move\pred,
 f(p_c)\move p_c, g(p_c)\move p_c\}.\end{align*} There is a critical pair $PC_1=(g(f(b))\to
 g(f(a)),\emptyset,\paire(q_{gfb},\pred))$ in $\autoInit\times\airr(R)$, which
 is resolved by adding transitions \begin{align*}
   a\move\paire(q_{N1},p_a) \\
   f(\paire(q_{N1},p_a))\move\paire(q_{N2},\pred) \\
   g(\paire(q_{N2},\pred))\move\paire(q_{N3},\pred) \\
   \paire(q_{N3},\pred)\move\paire(q_{gfb},\pred)
\end{align*} thereby producing automaton
 $\automate_1$. The supplementary operations do not create any new transition
 here. There is a critical pair $PC_2=(f(a)\to c,\emptyset,\paire(q_{N2},\pred))$ in
 $\automate_1$, which is resolved by adding transitions
 $c\move\paire(q_{N4},p_c)$ and $\paire(q_{N4},p_c) \move \paire(q_{N2},p_c)$, thereby producing
 automaton $\automate^\natural_2$. The supplementary operations are detailed
 further down and produce automaton $\automate_2$.
 Now consider that $g(c)\in\Rin (g(f(a)))$ and
 $g(f(a))\in\langrecs{\automate^\natural_2}{\paire(q_{gfb},\pred)}$ because we
 completed $PC_1$. But all what we have is
 $g(c)\moveun{\automate^\natural_2}g(\paire(q_{N4},p_c))\moveun{\automate^\natural_2}g(\paire(q_{N2},p_c))$,
 this last configuration being the left-hand side of no transition. As a result,
 $g(c)\notin\langrecs{\automate^\natural_2}{\paire(q_{gfb},p')}$ for any $p'$.
 
 The supplementary operations are made after completion of $PC_2$. Since there
 is a transition $g(\paire(q_{N2},\pred))\move\paire(q_{N3},\pred)$, we add a
 transition $g(\paire(q_{N2},p_c))\move\paire(q_{N3},p_{c})$. Then, since
 $q_{N3}\notin\autoInit$ and there is a
 $\paire(q_{N3},\pred)\move\paire(q_{gfb},\pred)$, we add
 $\paire(q_{N3},p_{c})\move\paire(q_{gfb},p_{c})$. No further transition needs
 to be added. These transitions allow
 $g(c)\in\langrecs{\automate_2}{\paire(q_{gfb}, p_{c})}$.
\end{rem}


\begin{lem}\label{lemme-correct}\label{lemme-eqcomp-preserve}
 Any automaton produced by innermost completion starting from some $\autoInit\times\airr(R)$ is correct w.r.t. $\Rin$.
\end{lem}
\begin{proof}
  Let $\automate$ be such an automaton; it is consistent with $\airr(R)$. This
  is Lemma~\ref{lemme-projdroite-irr}. Let $\paire(q,\pred)$ be a state of $\automate$, $u\in\langrecs\automate{\paire(q,\pred)}$ and $v\in \Rin(u)$. By definition of innermost rewriting, there is a rule $\ell\to r$ of $R$, a substitution $\mu:\ensvar\to\langtermesclos\ensfonc$ and a context $C$ such that $u=C[\ell\mu]$, $v=C[r\mu]$ and each strict subterm of $\ell\mu$ is a normal form. Let $u_0=\ell\mu$ and $v_0=r\mu$. There is a $\paire(q_0,\pred)$ such that $u_0\in\langrecs\automate{\paire(q_0,\pred)}$ and $C[\paire(q_0,\pred)]\moveet{\automate}{\paire(q,\pred)}$.
  
  Since $\ell$ is linear, there is a $\sigma:\ensvar\to Q_{\automate}$ such that $\ell\mu\moveet\automate\ell\sigma\moveet\automate\paire(q_0,\pred)$ and $r\mu\moveet\automate r\sigma$. This entails that $u\moveet\automate C[\ell\sigma]\moveet\automate C[\paire(q_0,\pred)]\moveet\automate \paire(q,\pred)$ and $v\moveet\automate C[r\sigma]$.
 Assume that 
 there is no $p_0$ such that $v_0\in\langrecs\automate{\paire(q_0,p_0)}$
 and show that $(\ell\to r, \sigma, \paire(q_0,\pred))$ is a critical pair in
 $\automate$. First, by assumption, there is no $p$ such that
 $r\sigma\moveet\automate\paire(q_0,p)$. Conditions~1 and~2 of
 Definition~\ref{def-cpn-eff} are thus met. Suppose
 $\ell=f(\gamma_1,\dots,\gamma_k)$ and show that condition~3 of
 Definition~\ref{def-cpn-eff} holds.\footnote{If $\ell$ is a constant, then
   condition~3 is vacuously true.} For each $i=1\ldots k$, let $\paire(q_i,p_i)$
 be the state of $\automate$ such that
 $\gamma_i\mu\moveet\automate\gamma_i\sigma\moveet\automate\paire(q_i,p_i)$ in
 the path of recognition of $\ell\sigma$. Then, by consistency with $\airr(R)$,
 for each $i=1 \ldots k$, $\gamma_i\mu\moveet{\airr(R)} p_i$. Since strict subterms of $\ell\mu$ are strict subterms of $u$ as well, they are normal forms, thus $p_i\neq\pred$, which validates condition~3 of Definition~\ref{def-cpn-eff}.
 
 Assume now that 
 $v_0\in\langrecs\automate{\paire(q_0,p_0)}$
 and show that 
 there is a $p$ such that $v\in\langrecs\automate{\paire(q,p)}$. This is obvious at the initial step $\autoInit\times\airr(R)$, and this property is conserved by completion as shown by Lemma~\ref{lemme-composition-contexte}.
\end{proof}

\noindent
A direct (practical) consequence of this result is that, if
completion terminates then there is no more critical pair and thus the tree
automaton is closed w.r.t. innermost rewriting. This used in the proof of the following theorem.
In the fixpoint tree automaton, $\automatefixin$, all states are products between states
of an automaton recognizing innermost reachable terms and states of
$\airr(R)$. Thus, $\projgauche\automatefixin$ recognizes all innermost
reachable terms and $\automatefixin$ recognizes all innermost reachable terms
{\em that are irreducible}, \ie normalized terms.




\begin{thm}[Correctness]
\label{thm-correction}
 Assuming $R$ is left-linear, the innermost equational completion procedure
 defined above produces a correct result whenever it terminates and produces
 some fixpoint $\automatefixin$ such that:
 \begin{equation}
  \langrec{\projgauche\automatefixin}\supset \Rin^*(\langrec{\autoInit}).\label{eq1}
 \end{equation}

 \begin{equation}
  \langrec{\automatefixin}\supset \Rin^!(\langrec{\autoInit}).\label{eq2}
 \end{equation}
\end{thm}

\begin{proof}
  Let $\automatefixin$ be the calculated fixpoint automaton. By
  Lemma~\ref{lemme-projdroite-irr}, $\automatefixin$ is consistent with
  $\airr(R)$, and therefore, by Lemma~\ref{lemme-eqcomp-preserve}
  and~\ref{lem-simp-correct}, $\automatefixin$ is correct w.r.t. $\Rin$. Since
  this automaton is a fixpoint, the case of Definition~\ref{defn-correct} where
  there remains a critical pair cannot occur, and therefore, for all states
  $\paire(q,\pred)$ of $\automate$, for all
  $u\in\langrecs\automatefixin{\paire(q,\pred)}$ and for all $v\in \Rin(u)$,
  there is a $p'$ such that $v\in\langrecs{\automatefixin}{\paire(q,p')}$.  The
  case where $u\in\langrecs\automatefixin{\paire(q,p)}$ with $p\neq \pred$ is
  not worth considering because $p\neq \pred$ means that $u$ is
  irreducible. Thus, $\langrecs\automatefixin{\paire(q,p)}$ necessarily contains
  all terms $\Rin$-reachable from $u$. Then, we first prove the~(\ref{eq2})
  case. Thanks to the previous results, for all terms $s,t$ such that $s\in
  \langrec\autoInit$ and $s\to_{\Rin}^* t$, we know that there exists a final
  state $q_f$ of $\autoInit$ and states $\pred,p'$ of $\airr(R)$ such that $s
  \in \langrecs{\automatefixin}{\paire(q_f,\pred)}$ and $t \in
  \langrecs{\automatefixin}{\paire(q_f,p')}$.  If $t$ is irreducible then
  $p'\neq \pred$ and thus $p'$ is a final state of $\airr(R)$. Since $q_f$ is
  final in $\autoInit$ and in $\airr(R)$, the state $\paire(q_f,p')$ is final in
  $\automatefixin$. Thus, $t$ is in $\langrec{\automatefixin}$.  For
  the~(\ref{eq1}) case, we use the same reasoning leading to the fact that if
  $s\in \langrec\autoInit$ and $s\to_{\Rin}^* t$, then $s \in
  \langrecs{\automatefixin}{\paire(q_f,\pred)}$ and $t \in
  \langrecs{\automatefixin}{\paire(q_f,p')}$. Then, we lift this property to
  $\projgauche\automatefixin$ using Lemma~\ref{lem-proj-eq}. Since the
  $\airr(R)$ component of $\automatefixin$ is complete,
  $\langrecs{\projgauche{\automatefixin}}{q_f} = \bigcup_{p\in P}
  \langrecs{\automatefixin}{\paire(q_f,p)}$. Thus,
  $\langrecs{\projgauche{\automatefixin}}{q_f}$ contains, in particular,
  $\langrecs{\automatefixin}{\paire(q_f,\pred)}$ that contains $s$ and
  $\langrecs{\automatefixin}{\paire(q_f,p')}$ that contains $t$.
\end{proof}

\section{Precision theorem}\label{sect-precision}

We just showed that the approximation is correct. Now we investigate its
accuracy on a theoretical point of view. This theorem is technical and difficult
to prove. However, this theorem is important because producing an   
over-approximation of reachable terms is easy (the tree automaton recognising
$\TF$ is a correct over-approximation) but producing an accurate approximation
is hard. To the best of our knowledge, no other work dealing with abstract
interpretation of functional programs or computing approximations of regular
languages can provide such a formal precision guarantee
(except~\cite{GenetR-JSC10} but in the case of general rewriting). Like
in~\cite{GenetR-JSC10}, we formally quantify the accuracy w.r.t. rewriting
modulo $E$, replaced here by {\em innermost} rewriting modulo $E$. The relation
of innermost rewriting modulo $E$, denoted by $\to_{\Rin/E}$, is defined as
rewriting modulo $E$ where $\to_{\Rin}$ replaces $\to_{R}$. We also define
$(\Rin/E)(L)$ and $(\Rin/E)^*(L)$ in the same way as $(R/E)(L)$, $(R/E)^*(L)$ where
$\to_{\Rin/E}$ replaces $\to_{R/E}$. 

The objective of the proof is to show that the completed tree automaton
recognises no more terms than those reachable by $\Rin/E$ rewriting.
The accuracy relies on the $\Rin/E$-coherence
property of the completed tree automaton, defined below. Roughly, a tree automaton $\A$ is
$\Rin/E$-coherent if $\moveet\A$ is coherent w.r.t. $R$ innermost rewriting
steps and $E$ equational steps. More precisely if $s \moveet\A q$ and $t
\moveet\A q$ with no epsilon transitions with color $\colour R$, then $s\equiv_E t$
(this is called separation of $E$-classes for $\A\sseps$).
And, if $t \moveet\A q$ with at least one epsilon transition with color
$\colour R$, then $s \to_{\Rin/E}^* t$ (this is called $\Rin$-coherence of
$\A$). Roughly, a tree automaton separates $E$-classes if all terms
recognized by a state are $E$-equivalent. Later, we will require this property on
$\A_0$ and then propagate it on $\A_i\sseps$, for all completed automata $\A_i$.

\def\classeeq#1#2{{\left[#2\right]}_{#1}}
\def\classeauto#1#2#3{{\left[#3\right]}_{#2}^{#1}}
\begin{defi}[Separation of $E$-classes]
\label{def-sep-classes}
 The pair automaton $\A$ separates the classes of $E$ if for any $q\in\projgauche{Q_\A}$,
there is a term $s$ such that for all $p\in\projdroite{Q_\A}$, $\langrecs\A{\paire(q,p)}\subset\classeeq Es$.
We denote by $\classeauto\A E q$ the class of terms in
$\langrecs\A{\paire(q,\cdot)}$, and extend this to configurations. We say that the
separation of classes by $\A$ is total if $\projgauche\A$ is accessible.
\end{defi}

\noindent
In the following, $\classeauto{\A\sseps} E q$ thus denotes the equivalence class
$[s]_E$, where $s$ is any term such that $s \moveet{\A\sseps} q$.

\begin{defi}[$\Rin/E$-coherence]
\label{rinco}
 An automaton $\A$ is $\Rin/E$-coherent if
 \begin{enumerate}
  \item $\A\sseps$ totally separates the classes of $E$,
  \item $\A$ is accessible, and
  \item for any state $\paire(q,p)$ of $\A$, $\langrecs\A{\paire(q,p)}\subset (\Rin/E)^*\left(\classeauto{\A\sseps}Eq\right)$.
 \end{enumerate}
\end{defi}

\noindent
Then, the objective is to show that the two basic elements of innermost
equational completion: completing a critical pair and applying an equation
preserve $\Rin/E$-coherence. We start by critical pair completion.
The first lemma shows that, if $(\ell\to r,\sigma,\paire(q,\pred))$ is a
critical pair of $\A$, then adding new (normalised) transitions preserves
$\Rin/E$-coherence.

\begin{lem}[Normalisation preserves $\Rin/E$-coherence]\label{lemme-norm-pres}
  Let $\A$ be a $\Rin/E$-coherent automaton, $\Q$ its set of states, $c\in\TFQ$, and $q$ a new state for
  $\A$. 
  The automaton $\A$ completed with new states and 
  transitions of $\normalisation{\A}(c,q)$ is
  $\Rin/E$-coherent. 
\end{lem}
\begin{proof}
   We prove that $\A \cup \normalisation{\A}(c,q)$ is $\Rin/E$-coherent by
   induction on the height of $c$.  
   \begin{itemize}
   \item 
   In the base case, $c$ is of height one. Thus, the value of $\normalisation{\A}(c,q)$ has necessarily been computed using
   case~(1) of Definition~\ref{def:normalisation}. Using this definition, we
   know that $c$ is of the form $f(q_1, \ldots, q_n)$, where $q_1, \ldots, q_n$ are
   states of $\A$ and $q$ is a new state for $\A$. Thus, what we need to show is
   that $\A \cup \{f(q_1, \ldots, q_n) \rw q\}$ is $\Rin/E$-coherent. 
   Let us denote by $\B$ the automaton $\A \cup \{f(q_1, \ldots, q_n) \rw
   q\}$. 
   Let us denote by $\Q_\A$ the set of states of $\A$.
   Since $q$ is a new state, adding this unique transition to $\A$ preserves the
   language recognized by all the states of $\Q_\A$ in $\B$. Thus, for any state $q'\in
   \B$ such that $q'\neq q$ the $\Rin/E$-coherence property is preserved: the
   language recognized by $\B\sseps$ in $\paire(\projgauche{q'},p)$ for all
   $p\in\projdroite{\Q_\A}$ is still a subset of $\classeeq E{s'}$ for
   some term $s'$, $q'$ remains accessible, and the language recognized by $q'$
   in $\B$ is still a subset of $(\Rin/E)^*\left(\classeauto{\B\sseps}E{q'}\right)$.
   What remains to be shown is that all the cases of Definition~\ref{rinco} are
   true for the state $q$. We start by showing case~(2) of this definition. The
   language $\langrecs{\B}{q}$ is not empty because $f(q_1,\ldots,q_n)\rw
   q \in \B$ and states $q_1, \ldots, q_n\in\A$ are different from $q$, and thus,
   accessible. Let us now show Case~(1). By definition of new states in normalisation,
   $q=\paire(r,p)$ where $r$ is a new state and $p\in\projdroite{\Q_A}$. For all $q_i$,
   $i=1\ldots n$, let $q_i=\paire(r_i,p_i)$ where $r_i\in \projgauche{\Q_A}$ and
   $p_i\in \projdroite{\Q_A}$. As mentioned above, we know that
   $\B\sseps$ totally separates the classes of~$E$ for all states of $\Q_\A$. We
   thus know that 
   for $r_i$, $i=1\ldots n$, there exist terms $s_i$ s.t. for all
   $p\in\projdroite{\Q_\A}$, $\langrecs{\B\sseps}{\paire(r_i,p)}\subseteq
   \classeeq E{s_i}$. Thus, for all $p\in\projdroite{\Q_\A}$, $\langrecs{\B\sseps}{\paire(r,p)}\subseteq \classeeq E{f(s_1,
     \ldots, s_n)}$. We can do a similar reasoning for case~(3). We know from
   above that the property holds for all $q_i = \paire(r_i,p_i)$, $i=1\ldots n$:
   $\langrecs{\B}{\paire(r_i,p_i)}\subseteq
   (\Rin/E)^*\left(\classeauto{\B\sseps}E{r_i}\right)$. Thus,
   $\langrecs{\B}{\paire(r,p)}= \langrecs{\B}{\paire(f(r_1,\ldots,r_n),p)} \subseteq
   (\Rin/E)^*\left(\classeauto{\B\sseps}E{r}\right)$.

   \item For the inductive case, we assume that the property is true for all
     configurations having a height inferior to the one of $c$. Since height of
     $c$ is greater than one, $\normalisation{\A}(c,q)$ can only be
     processed by case~(2) of Definition~\ref{def:normalisation}. Thus $c$ is of
     the form $C[f(q_1, \ldots, q_n)]$ and what we have to show is that 
     $\A \cup \normalisation{\A}(C[f(q_1, \ldots, q_n),q)$ is
     $\Rin/E$-coherent, {\em i.e.} that $\A \cup \{f(q_1,\ldots,q_n)\rw q'\}
     \cup \normalisation{\A\cup\{f(q_1,\ldots,q_n)\rw q'\}}(C[q'],q)$ is
     $\Rin/E$-coherent. Note that the height of $C[q']$ is strictly smaller to
     the height of $c$, and that all states of $C[q']$ belong to
     $\A\cup\{f(q_1,\ldots,q_n)\rw q'\}$. To apply the induction hypothesis on
     $\A \cup \{f(q_1,\ldots,q_n)\rw q'\} \cup
     \normalisation{\A\cup\{f(q_1,\ldots,q_n)\rw q'\}}(C[q'],q)$, what remains
     to prove is that $\A\cup\{f(q_1,\ldots,q_n)\rw q'\}$ is, itself,
     $\Rin/E$-coherent. For the case where $f(q_1,\ldots,q_n)\rw q'\in\A$ this
     is true because $\A\cup\{f(q_1,\ldots,q_n)\rw q'\}=\A$ and $\A$ is
     $\Rin/E$-coherent by assumption. Otherwise, the state $q'$ is new and the
     proof is exactly the same as in the base case. Finally, we can apply the
     induction hypothesis on $\A \cup \{f(q_1,\ldots,q_n)\rw q'\} \cup
     \normalisation{\A\cup\{f(q_1,\ldots,q_n)\rw q'\}}(C[q'],q)$ that entails
     the result.\qedhere
   \end{itemize}
 \end{proof}

\noindent
The second lemma shows that, if $(\ell\to
r,\sigma,\paire(q,\pred))$ is a critical pair of a tree automaton $\A$, if the
result of $\normalisation{\A}(r\sigma,\paire(q,\pred))$ is already part
of $\A$ and if $\A$ is $\Rin/E$-coherent, then adding
$\paire(q',p_{r\sigma})\movecol R\paire(q,p_{r\sigma})$ to $\A$ preserves
$\Rin/E$-coherence. 

\begin{lem}\label{lemme-heureka1}
  Let $\A$ be a $\Rin/E$-coherent automaton that is consistent with $\airr(R)$,
  let $(\ell\to r,\sigma,\paire(q,\pred))$ be a critical pair that is to be
  completed by adding transition $\paire(q',p_{r\sigma})\movecol
  R\paire(q,p_{r\sigma})$. We suppose that the normalisation steps have just
  been performed and still note $\A$ the resulting automaton. We have
  $\langrecs\A{r\sigma}\subset(\Rin/E)^*\left(\classeauto{\A\sseps} E q\right)$.
\end{lem}
\begin{proof}
  Let $t\in\langtermesclos\ensfonc$ such that $t\moveet\A
  r\sigma\moveet{\A\sseps}\paire(q',p_{r\sigma})$. Let $\mu$ be a substitution
  $\mu: \ensvar \to \TF$ such that $t=r\mu$ and $r\mu \moveet{\A} r\sigma$.
For each variable $x$ of $r$,
  let $x\mu$ be the subterm of $t$ at the position where $x$ occurs in $r$. For
  each variable $y$ appearing in $\ell$ but not in $r$, 
  let $y\mu$ be any term in $\langrecs\A{y\sigma}$. Since $\A$ is consistent
  with $\airr(R)$ and the critical pair fulfills Condition~3 of
  Definition~\ref{def-cpn-eff}, each strict subterm of $\ell\mu$ is a normal
  form. So $t=r\mu\in\Rin(\ell\mu)$. Moreover,
  $\ell\mu\in\langrecs\A{\paire(q,\pred)}$ and $\A$ is $\Rin/E$-coherent, so
  $t\in(\Rin/E)^*\left(\classeauto{\A\sseps} E q\right)$.
\end{proof}

\begin{lem}[Supplementary transitions preserve $\Rin/E$-coherence]
\label{lem-supop}
Let $\A$ be a $\Rin/E$-coherent automaton that is consistent with $\airr(R)$.
Assume that $\A$ has been completed with the transitions to have $s
\moveet\A \paire(q,p)$ and the necessary supplementary transitions. Let $C[\,]$ be a context and $\paire(q_c,p_c)$ a 
state of $\A$ such that $C[s]\moveet\A \paire(q_c,p_c)$. If 
$\langrecs\A{s}\subset(\Rin/E)^*\left(\classeauto{\A\sseps} E
  {q}\right)$ then 
$\langrecs\A{C[s]}\subset(\Rin/E)^*\left(\classeauto{\A\sseps} E {q_c}\right)$.
\end{lem}
\begin{proof}
  We make a proof by induction on the height of context $C[\,]$. If $C[\,]$ is
  of height $0$ then the result holds because
  $\langrecs\A{s}\subset(\Rin/E)^*\left(\classeauto{\A\sseps} E
    {q}\right)$ is an assumption of the lemma. Assume that the result holds for
  contexts of height strictly smaller to $k$. Let $C[\,]$ be a context of height
  $k$. Let $f$ be the symbol of arity $n$, $C'[\,]$ the context of height $k-1$,
  $1\leq i \leq n$ a natural and $t_1,\ldots,t_n$ the
  terms such that $t_i=s$ and $C[s] =
  C'[f(t_1,\ldots,t_{i-1},s,t_{i+1},\ldots,t_n)$. Since
  $C[s]\moveet\A \paire(q_c,p_c)$, we know that there exists states
  $\paire(q_j,p_j)$, $1\leq j \leq n$ such that $t_j \moveet\A \paire(q_j,p_j)$
  for all $1\leq j \leq n$ with $q_i=q$ and $p_i=p_s$, since $t_i=s$. Furthermore, we know that there exists a state
  $\paire(q',p')$ and a transition
  $f(\paire(q_1,p_1),\ldots,\paire(q_{i-1},p_{i-1}),\paire(q,p_s),\paire(q_{i+1},p_{i+1}),
  \ldots, \paire(q_n,p_n)) \move \paire(q',p')$ in $\A$ such that
  $C'[\paire(q',p')] \moveet\A \paire(q_c,p_c)$. If this transition belongs to
  $\A$ before solving the critical pair, then we can conclude using the
  $\Rin/E$-coherence of the initial $\A$. If the transition is a supplementary
  transition, this means that there exists a transition
  $f(\paire(q_1,p_1),\ldots,\paire(q_{i-1},p_{i-1}),\paire(q,\pred),\paire(q_{i+1},p_{i+1}),
  \ldots, \paire(q_n,p_n)) \move \paire(q',\pred)$ in the initial tree automaton
  which is $\Rin/E$-coherent. So, initially,
  $\langrecs\A{q'}\subset(\Rin/E)^*\left(\classeauto{\A\sseps} E
    {q'}\right)$. Adding the transition where $\paire(q,p_s)$ replaces
  $\paire(q,p_{red})$ preserves this property on $q'$ because we know by
  assumption that
  $\langrecs\A{s}\subset(\Rin/E)^*\left(\classeauto{\A\sseps} E
    {q}\right)$. As $t_i=s$, we get that
  $\langrecs\A{f(t_1,\ldots,t_n)}\subset(\Rin/E)^*\left(\classeauto{\A\sseps} E {q'}\right)$
  remains true with supplementary transitions. Finally, we can use the induction
  hypothesis on $C'[\,]$ that is of height $k-1$ to get that
  $\langrecs\A{C'[f(t_1,\ldots,t_n)]}\subset(\Rin/E)^*\left(\classeauto{\A\sseps} E
    {q_c}\right)$.
\end{proof}

\noindent
The three above lemmas can straightforwardly be lifted to the full completion
algorithm as follows.
\begin{lem}\label{l1}
 Completion of an innermost critical pair preserves $\Rin/E$-coherence.
\end{lem}
\begin{proof}
  Simple combination of Lemmas~\ref{lemme-norm-pres}, \ref{lemme-heureka1} and~\ref{lem-supop}.
\end{proof}

\noindent
The next theorem aims at showing that applying an equation preserves
$\Rin/E$-coherence. We first prove a lemma showing that the equivalence classes associated to
two states concerned by an equation application are equal.

\begin{lem}\label{lemme-precision-equation}
  Let $\A$ be an automaton that totally separates the classes of $E$. Let
  $(s=t,\theta,\paire(q_1,p),\paire(q_2,p))$ be a situation of application of an
  equation of $E$ in $\A$. Then $\classeauto\A E{q_1}=\classeauto\A E{q_2}$.
\end{lem}

\begin{proof}
  It suffices to prove that $\classeauto\A E{s\theta}=\classeauto\A
  E{t\theta}$. Since the separation of the classes by $\A$ is total, for each
  $x$ in the domain of $\theta$, there is a term
  $x\mu\in\langrecs{\projgauche\A}{x\theta}$. This builds an instance
  $s\mu\equiv_Et\mu$ of the considered equation. But $\classeauto\A
  E{s\theta}=\classeeq E{s\mu}$ and $\classeauto\A E{t\theta}=\classeeq
  E{t\mu}$.  
\end{proof}

\begin{thm}\label{l2}
 Equational simplification preserves $\Rin/E$-coherence.
\end{thm}
\begin{proof}
 Let $s=t\in E$, and $(s=t,\theta,\paire(q_1,p_0),\paire(q_2,p_0))$ be a situation of
 application of this equation in $\A$. Let $\B$ be the automaton
 resulting from the application of the equation between $\paire(q_1,p_0)$ and $\paire(q_2,p_0)$. Let $\paire(q,p)$ be a state. 
 
 Show that $\langrecs{\B\sseps}{\paire(q,p)}\subset\classeauto{\A\sseps}E q$. Consider $u=C[u_0]\moveet{\A\sseps}C[\paire(q_1,p_0)]\movecol E C[\paire(q_2,p_0)] \moveet{\A\sseps} \paire(q,p)$. We have $u_0\in\classeauto{\A\sseps}E{q_1}$, thus, by Lemma~\ref{lemme-precision-equation}, $u_0\in\classeauto{\A\sseps}E{q_2}$. Therefore $u\in C[\classeauto{\A\sseps}E{q_2}]$, which is just $\classeauto{\A\sseps}Eq$. Other cases are either trivial, symmetrical or reducible to this one.
 
 Next, show that $\langrecs{\B}{\paire(q,p)}\subset(\Rin/E)^*\left(\classeauto{\A\sseps}Eq\right)$. Consider $u=C[u_0]\moveet{\A}C[\paire(q_1,p_0)]\movecol E C[\paire(q_2,p_0)] \moveet{\A} \paire(q,p)$. We have $u_0\in(\Rin/E)^*\left(\classeauto {\A\sseps}E{q_1}\right)$, \ie $u_0\in(\Rin/E)^*\left(\classeauto {\A\sseps}E{q_2}\right)$. Thus $u\in(\Rin/E)^*\left(C[\classeauto {\A\sseps}E{q_2}]\right)$.
 First, assume that $C[\paire(q_2,p_0)] \moveet{\A\sseps} \paire(q,p)$. Then $C[\classeauto{\A\sseps}E{q_2}]=\classeauto{\A\sseps}E q$, therefore $u\in(\Rin/E)^*\left(\classeauto {\A\sseps}E{q}\right)$.
 Second, assume that $C[\paire(q_2,p_0)] \moveet{\A} \paire(q,p)$ with just one
 $\colour R$-transition, that is
 $C[\paire(q_2,p_0)]\moveet{\A\sseps}D[r\sigma]\moveet{\A\sseps}D[\paire(q'_3,p_3)]\movecol
 R D[\paire(q_3,p_3)]\moveet{\A\sseps}\paire(q,p)$. There is a corresponding
 critical pair $(\ell\to r,\sigma,\paire(q_3,\pred))$ and, by
 Lemma~\ref{lemme-heureka1},
 $\langrecs\A{r\sigma}\subset(\Rin/E)^*\left(\classeauto{\A\sseps} E
   {q_3}\right)$. On the other hand,
 $D[\classeauto{\A\sseps}E{q_3}]=\classeauto{\A\sseps}Eq$. Thus, we have $\langrecs\A{D[r\sigma]}\subset(\Rin/E)^*\left(\classeauto{\A\sseps} E {q}\right)$. Since $(\Rin/E)^*$ is an operator that deals with equivalence classes, every term equivalent to one of $\langrecs\A{D[r\sigma]}$ is also a descendant of $\classeauto{\A\sseps}E q$ by $\Rin/E$. Since $C[\classeauto{\A\sseps}E {q_2}]=D[\classeauto{\A\sseps} E {r\sigma}]$, $u$ is a descendant of such a term, so $u\in(\Rin/E)^*\left(\classeauto {\A\sseps}E{q}\right)$.
 
 Finally, in the paragraph above, it suffices that $D[\classeauto{\A\sseps}E{q_3}]\subset(\Rin/E)^*\left(\classeauto{\A\sseps}Eq\right)$ (we had $D[\classeauto{\A\sseps}E{q_3}]=\classeauto{\A\sseps}Eq$): this allows us to reuse this case as an induction step over the number of $\colour R$-transitions present in the path $C[\paire(q_2,p_0)] \moveet{\A} \paire(q,p)$.
\end{proof}

\noindent
To state the final theorem, we need a last lemma guaranteeing that equivalence
classes associated to states in $\A_i\sseps$ do not change during completion.

\begin{lem}[Completion preserves equivalence classes]\label{lemme-class-pres} Let $\A_0$ be a pair automaton,
  $q$ be a state of $\projgauche{\A_0}$ and $\A_i$ be an automaton obtained after
  some completion steps of $\A_0$. If $\A_0\sseps$ and $\A_i\sseps$ totally
  separate the classes of $E$ then
  $\classeauto{{\A_0}\sseps}E{q} = \classeauto{{\A_i}\sseps}E{q}$.
\end{lem}
\begin{proof}
  Since $q$ is a state of $\projgauche{\A_0}$ and ${\A_0}\sseps$ totally
  separates the classes of $E$, we know by Definition~\ref{def-sep-classes} that
  there is a term $s$ such that for all $p$,
  $\langrecs{{\A_0}\sseps}{\paire(q,p)}\subset\classeeq Es$. Let $t$ be a term
  belonging to $\classeeq Es$ and to $\langrecs{\A_0\sseps}{\paire(q,p')}$, for
  a given $p'$. We
  know that such a term exists because ${\A_0}\sseps$ {\em totally} separates
  the classes, {\em i.e.} ${\A_0}\sseps$ is accessible. Besides, since
  $\A_i\sseps$ also totally separates the classes of 
  $E$, we know from Definition~\ref{def-sep-classes} that there is a term $u$
  such that for all $p$, $\langrecs{{\A_i}\sseps}{\paire(q,p)}\subset\classeeq
  Eu$. However, since completion only adds transitions to automata, we know that
  the term $t$ is also recognized by $\A_i\sseps$, {\em i.e.} $t\in
  \langrecs{{\A_i}\sseps}{\paire(q,p')}$. Thus $t$ belongs to $\classeeq Eu$ and
  since $t$ also belong to $\classeeq Es$, we have $s\equiv_Eu$. Finally $\classeeq Es = \classeeq
  Eu$ and thus $\classeauto{{\A_0}\sseps}E{q} = \classeauto{{\A_i}\sseps}E{q}$.
\end{proof}

\noindent
Finally, if $\A_0$ separates the classes of $E$,
innermost equational completion will never add to the computed
approximation a term that is not a descendant of $\langrec{\A_0}$ through $\Rin$
modulo $E$ rewriting. This is what is stated in this main theorem, which formally
defines the precision of completed tree automata.

\begin{thm}[Precision]\label{thm-precision}
 Let $E$ be a set of equations. Let $\A_0=\autoInit\times\airr(R)$, where $\autoInit$ has designated final states. We prune $\A_0$ of its non-accessible states. Suppose $\A_0$ separates the classes of $E$. Let $R$ be any left-linear TRS. Let $\A_i$ be obtained from $\A_0$ after some steps of innermost equational completion. Then
  \begin{equation}
  \langrec{\projgauche{\A_i}}\subset (\Rin/E)^*(\langrec{\autoInit})). \label{p1}
 \end{equation}

  \begin{equation}
  \langrec{\A_i}\subset (\Rin/E)^!(\langrec{\autoInit})). \label{p2}
 \end{equation}

\end{thm}

\begin{proof} 
  We know that $\A_0$ is $\Rin/E$-coherent because (1) $\A_0\sseps$ separates
  the classes of $E$ ($\A_0$ separates the classes of $E$ and $\A_0=\A_0\sseps$
  since none of $\A_{init}$ and $\airr$ have epsilon transitions), and (2)
  $\A_0$ is accessible. Condition (3) of Definition~\ref{rinco} is trivially
  satisfied since $\A_0$ separates classes of $E$, meaning that for all states
  $q$, there is a term $s$ s.t. $\langrecs{\A_0}{\paire(q,p)}\subset\classeeq
  Es$, {\em i.e.} all terms recognized by $q$ are $E$-equivalent to $s$ which is
  a particular case of case (3) in Definition~\ref{rinco}. Then, during
  successive completion steps, by Lemma~\ref{l1} and~\ref{l2}, we know that each
  basic transformation applied on $\A_0$ (completion or equational step)
  preserves the $\Rin/E$-coherence of $\A_0$. Thus, all automata from $\A_0$ to
  $\A_i$ are $\Rin/E$-coherent. Finally, case (3) of $\Rin/E$-coherence of $\A_i$ entails
  that, for all states $\paire(q,p)$ of $\A_i$,
  $\langrecs{\A_i}{\paire(q,p)}\subset
  (\Rin/E)^*\left(\classeauto{{\A_i}\sseps}Eq\right)$. Since completion does not
  add final states, for any final state $\paire(q_f,p_f)$ of $\A_i$,
  we know that $\paire(q_f,p_f)$ is a final state of $\A_0$, and that $q_f$ is
  a final state of $\autoInit$. From the fact that $\paire(q_f,p_f)$
  is a state of $\A_0$ and $\A_0$ is
  accessible, we can deduce that $\classeauto{{\A_0}\sseps}E{q_f}$ is not empty.
  Then, using Lemma~\ref{lemme-class-pres}, we can obtain that $\classeauto{{\A_0}\sseps}E{q_f} = 
  \classeauto{{\A_i}\sseps}E{q_f}$. Besides, since neither $\autoInit$ nor
  $\airr(R)$ have epsilon transitions, ${\A_0}\sseps= \A_0$. Thus, in the above
  inequality, we can choose $q_f$ for $q$ and replace
  $\classeauto{{\A_i}\sseps}E{q_f}$ by $\classeauto{{\A_0}}E{q_f}$. Suming-up we get
  that, for all final state $q_f$ of $\autoInit$, for all state $p$ of $\airr(R)$, $\langrecs{\A_i}{\paire(q_f,p)}\subset
  (\Rin/E)^*\left(\classeauto{\A_0}E{q_f}\right)$. Then using
  Lemma~\ref{lem-proj-eq} and the fact that $\airr(R)$ is complete, we obtain
  that for any final state $q_f$ of $\autoInit$, $\langrecs{\projgauche{\A_i}}{q_f} = \bigcup_p
  \langrecs{\A_i}{\paire(q_f,p)} \subset
  (\Rin/E)^*\left(\classeauto{\A_0}E{q_f}\right)$. Note that, for any $E$-equivalence
  class $[s]_E$, $(\Rin/E)^*([s]_E) = (\Rin/E)^*(\{s\})$. Thus, we can simplify
  the above inequality into $\langrecs{\projgauche{\A_i}}{q_f} \subset
  (\Rin/E)^*\left(\bigcup_p \langrecs{\A_0}{\paire(q_f,p)})\right)$. Using again
  Lemma~\ref{lem-proj-eq} on $\bigcup_p \langrecs{\A_0}{\paire(q_f,p)}$ and the
  fact that $\projgauche{\A_0}=\autoInit$, we get
  that for any final state $q_f$ of $\autoInit$, we have $\langrecs{\projgauche{\A_i}}{q_f} \subset
  (\Rin/E)^*\left(\langrecs{\autoInit}{q_f})\right)$ and finally $\langrec{\projgauche{\A_i}} \subset
  (\Rin/E)^*\left(\langrec{\autoInit})\right)$. This ends the proof of
  case~(\ref{p1}). Case~(\ref{p2}) can be shown using again the property proved
  above, $\langrecs{\A_i}{\paire(q_f,p)}\subset
  (\Rin/E)^*\left(\classeauto{\A_0}E{q_f}\right)$, and by remarking that
  for $\paire(q_f,p)$ to be a final state of $\A_i$ we need $p\neq \pred$. Thus,
  terms recognized by $\paire(q_f,p)$ are innermost reachable modulo $E$ {\em
    and} irreducible.
\end{proof}

\noindent
Note that the fact that $\A_0$ needs to separate the classes of $E$ is not a
strong restriction in practice. In the particular case of 
functional TRS (TRS encoding first order typed functional
programs~\cite{Genet-WRLA14,Genet-JLAMP15}), $E$ is non empty and is inferred
from $R$ (see an example Section~\ref{equation}). In this case, there always exists a tree automaton recognising a
language equal to $\Lang(\A_0)$ and which separates the classes of $E$,
see~\cite{Genet-rep14} for details. However, outside of this particular case,
this is not true in general.
For instance, if $E=\emptyset$ then for $\A_0$ to separate the classes of $E$,
it needs to recognize a finite language. Indeed, if $E=\emptyset$ then for all terms
$t$, $\classeeq Et$ is a singleton, {\em i.e.}  $\classeeq Et= \{t\}$. For
$\A_0$ to separate the classes of $E$, for all states $\paire(q,p)$ of $\A_0$ we need to
have a term $s$ such that $\langrecs{\A_0}{\paire(q,p)} \subseteq \classeeq Es =
\{s\}$. Thus, language of $\A_0$ is necessarily finite. For the particular case
where $E=\emptyset$, we have a specific corollary of the above theorems.
\begin{cor}
\label{corollary-emptyset}
Let $R$ be a left-linear TRS, $E$ an empty set of equations, and
$\A_0=\autoInit\times\airr(R)$. If $\A_0$ recognizes a finite language and 
innermost completion terminates on a fixpoint $\automatefixin$ then:
 \begin{equation}
  \langrec{\projgauche{\automatefixin}} = \Rin^*(\langrec{\autoInit}))
 \end{equation}

  \begin{equation}
  \langrec{\automatefixin} = \Rin^!(\langrec{\autoInit})). 
 \end{equation}
\end{cor}
\begin{proof}
This is a consequence of Theorems~\ref{thm-correction} and of
Theorem~\ref{thm-precision} where $E=\emptyset$ and $\A_0$ trivially separates the classes
of $E$. 
\end{proof}

Note that if $\Lang(\A_0)$ is not finite, it is possible to separate classes of
$E$ but at the price of a complete transformation of $R$ and $\A_0$. The new
initial automaton recognizes a single constant term, say $S$. This new automaton
trivially separates the classes of $E$, for any $E$ (empty or not). Besides, we
add to $R$ the necessary rewrite rules (grammar rules in fact) to generate
$\Lang(\A_0)$ by rewriting $S$ (the axiom of the
grammar). See~\cite{Genet-rep14} for details and~\cite{Genet-JLAMP15} to see
how this has been used to show that standard completion {\em exactly} computes 
reachable terms when $E=\emptyset$.

\section{Improving accuracy of static analysis of functional programs}
\label{sect-fun}
We just showed the accuracy of the approximation on a theoretical side. Now we
investigate the accuracy on a practical point of view.
There is a recent and renewed interest for Data flow analysis of higher-order
functional
programs~\cite{KobayashiTU-POPL10,OngR-POPL11,KochemsO-RTA11,KobayashiI-ESOP13}
that      
was initiated by~\cite{JonesA-TCS07}. 
None of those techniques is strategy-aware: on
Example~\ref{ex:completionInnermost}, they all consider the term
$c(a(s(0)),f(n))$ as reachable with innermost strategy, though it is not.
Example~\ref{ex:completionInnermost} also shows that this is not the case with
innermost completion.  

We made an alpha implementation of innermost equational completion. This new version of
\timbuk, named \timbuk STRAT, is available at~\cite{timbuk} along with
several examples. On those examples, innermost equational completion runs within
milliseconds. Sets of approximation equations, when needed, are systematically
defined using~\cite{Genet-WRLA14,Genet-JLAMP15}. More details about this step
can be found in Section~\ref{equation}. They are used to guarantee termination of 
completion. 
Now, we show that accuracy of innermost equational completion
can benefit to static analysis of functional 
programs. As soon as one of the analyzed functions is not terminating
(intentionally or because of a bug), not taking the evaluation strategy into
account may result into an imprecise analysis. Consider the following OCaml
program:

\medskip
{\footnotesize
\begin{lstlisting}[numbers=none]
let hd= function x::_ -> x;;          let tl= function _::r -> r;;
let rec delete e ls=
 if (ls=[]) then [] else 
    if (hd ls=e) then (tl ls) else (hd ls)::(delete e ls);;
\end{lstlisting}
} 

\medskip
\noindent
It is faulty: the recursive call should be \lstinline+(hd l)::(delete e (tl ls))+.
Because of this error, any call \lstinline+(delete e ls)+ will not terminate if \lstinline+ls+ is
not empty and \lstinline+(hd ls)+ is not \lstinline+e+. We can encode the above
program into a TRS.
Furthermore, if we consider only two elements in lists
({\tt a} and {\tt b}), the language $L$ of calls to \lstinline+(delete a l)+,
where \lstinline+l+ is any non empty list of {\tt b}, is regular. 
Thus, 
standard completion can compute an automaton 
over-approximating $R^*(L)$. Besides, the automaton $\airr(R)$ recognising
normal forms of $R$ can be computed since $R$ is left-linear.
Then, by computing the intersection between the two
automata, we obtain the automaton recognising an over-approximation of the set
of reachable terms in normal form\footnote{Computing $\airr(R)$ and the
  intersection can be done using \timbuk.}, {\em i.e.} normalized terms. Assume that we have an abstract
OCaml interpreter performing completion and intersection with $\airr(R)$:

\medskip

\noindent
{\small
\textcolor{OliveGreen}{\tt \# delete a [b+];;}\\
\textcolor{OliveGreen}{\tt -:abst list= empty}
}

\medskip

\noindent
The \lstinline+empty+ result reflects the fact that the \lstinline+delete+
function does not compute any \emph{result}, \ie it is not terminating on all the given input
values. Thus the language of results is empty. Now, assume that we consider
calls like \lstinline+hd(delete e l)+. In 
this case, any analysis technique ignoring the call-by-value
evaluation strategy of OCaml will give imprecise results. This is due to the
fact that, for any non empty list \lstinline+l+ starting with an element
\lstinline+e'+ different from \lstinline+e+, \lstinline+(delete e l)+ rewrites
into \lstinline+e'::(delete e l)+, and so on. Thus \lstinline+hd(delete e l)+,
can be rewritten into \lstinline+e'+ with an outermost rewrite strategy. Thus, if we
use an abstract OCaml interpreter built on the standard completion, we will have
the following interaction:

\medskip

\noindent
{\small
\textcolor{OliveGreen}{\tt \# hd (delete a [b+]);;}\\
\textcolor{OliveGreen}{\tt -:abst list= b}
}

\medskip

\noindent
The result provided by the abstract interpreter is imprecise. It fails to reveal
the bug in the delete function since it totally hides the fact that the delete
function does not terminate! Using innermost equational completion and \timbuk STRAT on the
same example gives the expected result which is\footnote{see files 
  \texttt{nonTerm1} and \texttt{nonTerm1b} in the \timbuk STRAT distribution at~\cite{timbuk}.}: 

\medskip

\noindent
{\small
\textcolor{OliveGreen}{\tt \# hd (delete a [b+]);;}\\
\textcolor{OliveGreen}{\tt -:abst list= empty}
}

\medskip

\noindent 
We can perform the same kind of analysis for the program {\tt sum} given in the
introduction. This program does not terminate (for any input) with call-by-value, 
but it terminates with call-by-name strategy. 
Again, strategy-unaware methods cannot show this: there are
(outermost) reachable terms that are in normal form: the integer results
obtained with a call-by-need or lazy evaluation. 
An abstract OCaml interpreter unaware of strategies would say: 

\medskip

\noindent
{\small
\textcolor{OliveGreen}{\tt \# sum s*(0);;}\\
\textcolor{OliveGreen}{\tt -:abst nat= s*(0)}
}

\medskip

\noindent
where a more precise and satisfactory answer would be {\small
  \textcolor{OliveGreen}{\tt -:abst nat= empty}}. Using \timbuk STRAT, we can
get this answer (see Section~\ref{equation} for  the \timbuk STRAT input file). To over-approximate the set of results of the 
function {\tt sum} for all natural numbers {\tt i}, we can start innermost
equational completion with the initial regular language $\{sum(s^*(0))\}$. Let $\A= ( \F, \Q, \Q_f,
\Delta )$ with $\Q_f=\{q_1\}$ and $\Delta=\{ 0 \move q_0, s(q_0) \move q_0,
sum(q_0) \move q_1 \}$ be an automaton recognising this
language. Innermost equational completion with \timbuk STRAT terminates on an
automaton where the only product state labeled by
$q_1$ is $\paire(q_1,\pred)$. This means that terms of the form
$sum(s^*(0))$ have no innermost normal form, {\em i.e.} the function
{\tt sum} is {\em not terminating} with call-by-value for all input values. On
all those examples, we used initial automata $\A$ that were not separating
equivalences classes of $E$. On those particular examples the precision of innermost
completion was already sufficient for our verification purpose. 
Yet, if accuracy is not sufficient, it is possible to refine
$\A$ into an equivalent automaton separating equivalences classes of
$E$, see~\cite{Genet-rep14}. When necessary, this permits to exploit the full power
of the precision Theorem~\ref{thm-precision} and get 
an approximation of innermost reachable terms, as precise as possible, w.r.t. $E$.
The last example deals with higher-order functions. The following OCaml
defines the {\tt even} and {\tt odd} predicates, the {\tt length}
function  as well as the {\tt map} higher order function.

\begin{tabular}{l|l}
\begin{minipage}{7cm}
{\small
\begin{lstlisting}[escapechar=@,numbers=none]
let rec map f ls=
  match ls with
      [] -> []
    | e::r -> (f e)::(map f r);;

let rec length ls=
  match ls with
      [] -> 0
    | e::r -> 1+(length r);;
\end{lstlisting}
}
\end{minipage}   &
\begin{minipage}{5cm}
{\small
\begin{lstlisting}[escapechar=@,numbers=none]
let rec even x= 
  match x with
      0 -> true
    | _ -> odd (x - 1)
and odd x=
  match x with 
      0 -> false
    | _ -> even (x + 1);;
\end{lstlisting}
}
\end{minipage}
\end{tabular}

\noindent
The problem with the above program is that the \lstinline+odd+ predicate is not terminating on natural
numbers greater than 0. Thus, \lstinline+even+ is not terminating on natural numbers
greater than 1. Thus, calling \lstinline+(map even ls)+ on any list
\lstinline+ls+ containing at least one natural number greater than 1, is not
terminating. However, if we use an abstract OCaml interpreter that do not take
call-by value strategy, we will have the following behavior.

\medskip
\noindent
{\small
\textcolor{OliveGreen}{\tt \# map even [s(s(s*(0)))]} \\
\textcolor{OliveGreen}{\tt -:abst list= empty}\\
\textcolor{OliveGreen}{\tt \# length (map even [s(s(s*(0)))])} \\
\textcolor{OliveGreen}{\tt -:abst list= s(0)}\\
}

\noindent
This is due to the fact that \lstinline+map+ builds a list of the form
\lstinline+(even s(s(s*(0))))::[]+ whose length cannot be computed if we use
call-by-value strategy, but whose length is 1 with call-by-name strategy. Again,
such an interpretation yields very imprecise results if it does not take evaluation
strategies into account. We can deal with this example using innermost
completion. First, we have to encode higher-order functions into first order
terms. This can be done using the encoding of~\cite{Reynolds-IP69}: defined
symbols become constants, constructor symbols 
remain the same, and an additional {\em application} operator \verb|app| of
arity 2 is introduced. The above OCaml program is thus encoded in the following \timbuk\
TRS.
{\small
\begin{alltt}
\Ops map:0 length:0 even:0 odd:0 s:1 o:0 nil:0 cons:2 app:2 true:0 false:0
\Vars F X Y Z Xs

\TRS R1
    app(app(map,F),nil) -> nil
    app(app(map,F), cons(X,Y)) -> cons(app(F,X),app(app(map,F), Y))

    app(even, o) -> true 
    app(even, s(X)) -> app(odd,X)
    app(odd, o) -> false
    app(odd, s(X)) -> app(even,s(s(X)))

    app(length, nil) -> o
    app(length, cons(X,Y)) -> s(app(length,Y))
\end{alltt}
}

\noindent
Using the following tree automaton, we can define the language of terms of the
form {\tt app(length,app(app(map,even),ls))} where {\tt ls} is any list
containing at least one natural number greater than 1.

{\small
\begin{alltt}
\Automaton A0
\States qlen qmap qeven qf qmapeven qmapeven2 qnil ql ql2 ql3 q0 q1 qn
\FinalStates qf
\Transitions
 length -> qlen            map -> qmap                   even -> qeven   
 app(qlen,qmapeven2)->qf   app(qmapeven,ql)->qmapeven2   app(qmap,qeven)->qmapeven
 cons(q0,ql2) -> ql        cons(q1,ql2) -> ql            cons(qn,ql2) -> ql
 cons(q0,ql2) -> ql2       cons(q1,ql2) -> ql2           cons(qn,ql2) -> ql2
 cons(qn,ql3) -> ql2       cons(q0,ql3) -> ql3           cons(q1,ql3) -> ql3
 cons(qn,ql3) -> ql3       cons(q0,qnil) -> ql3          cons(q1,qnil) -> ql3
 cons(qn,qnil) -> ql3      nil -> qnil                   o -> q0
 s(q0) -> q1               s(q1) -> qn                   s(qn) -> qn
\end{alltt}
}
\noindent
Innermost completion of this automaton yields a tree automaton whose set of
irreducible (constructor) terms is empty, meaning that the set of possible
abstract results for this function call, using call-by-value, is empty.

On all the above examples, all aforementioned
techniques~\cite{OngR-POPL11,KochemsO-RTA11,
JonesA-TCS07}, as well as all standard completion techniques~\cite{TakaiKS-RTA00,
GenetR-JSC10,Lisitsa-RTA12}, give a more coarse
approximation and are unable to prove strong non-termination with
call-by-value. 
Indeed, those
techniques approximate all reachable terms, independently of the rewriting
strategy. Their approximation will, in particular, contain the integer results
that are reachable by the call-by-need evaluation strategy. 

\section{Inferring sets of equations}
\label{equation}
Sets of equations are inferred using the technique of~\cite{Genet-JLAMP15}.
We explain the application of this technique on the {\tt sum} example 
used in the introduction. We recall the OCaml program and give its associated TRS.

\noindent
{\small
\begin{tabular}{lll}
\begin{minipage}{.43\linewidth}
\begin{lstlisting}
let rec sumList x y= 
  (x+y)::(sumList (x+y) (y+1));;
\end{lstlisting}
\end{minipage} & ~ &
\begin{minipage}{.4\linewidth}
\begin{lstlisting}
let rec nth i (x::l)= 
  if i<=0 then x else nth (i-1) l;;
let sum x= nth x (sumList 0 0);;
\end{lstlisting}
\end{minipage}
\end{tabular}}

\noindent
{\small
$\begin{array}{lll}
(1)0 + X \rw X && (4)nth(0,cons(X,Y)) \rw X \\
(2)s(X) + Y \rw s(X+Y) && (5)nth(s(X),cons(Y,Z)) \rw nth(X,Z) \\
(3)sumList(X,Y) \rw cons(X+Y, sumList(X+Y,s(Y))) && (6)sum(X) \rw
nth(X,sumList(0,0)) 
\end{array}$}

\medskip
\noindent
TRS under consideration in~\cite{Genet-JLAMP15} are called "functional
TRS". They are typed TRS encoding typed functional programs of the ML family. For
sake of simplicity we omit the type information here.
The set of symbols $\F$ of the TRS is separated into two disjoint sets:
the set $\D$ (defined symbols) appearing on the top of a left-hand side of
a rule, and constructor symbols $\C=\F\setminus\D$.
On the {\tt sum} TRS, $\D=\{+,sumList,nth,sum\}$ and $C=\{0,s,cons,nil\}$.
The set of equations to use is $E=E_R \cup E^r \cup E^c$ where $E_R= \{\ell = r
\sep \ell \rw r \in R\}$, $E^r=\{f(x_1, \ldots, x_n)=f(x_1, \ldots, x_n) \sep
f \in \F, \mbox{ and arity of $f$ is $n$}\}$ and $E^c$ is a set of equations of
the form $u=u|_p$ where $u$ is a term built on $\C$ and $\X$. The sets $E_R$ and
$E^r$ are fixed by $R$, but $E^c$ can be adapted so as to tune the precision of the
approximation. Nevertheless, to have a terminating completion, $E^c$ has to fulfill
the following property: the set of equivalence classes of (well-typed) terms in $\TC$
w.r.t. $=_{E^c}$ has to be finite. On the {\tt sum} example a possible choice
for $E^c$ is $E^c=\{cons(x, cons(y, z)) = cons(x, z),$ $s(s(x))=s(x)\}$, such that
the set of equivalence classes of (well-typed) terms in $\TC$ is finite, {\em
  i.e.} it consists only of five equivalence classes: $0$, $s(0)$, $nil$, $cons(0,nil)$,
$cons(s(0),nil)$. Here is the complete \timbuk\ specification for this example.

{\small
\begin{alltt}
\Ops sum:1 nth:2 sumList:1 cons:2 nil:0 zero:0 s:1 add:2
\Vars X Y Z U
\TRS R1
add(zero,X) -> X                               nth(zero,cons(X,Y)) -> X
add(s(X),Y) -> s(add(X,Y))                     nth(s(X),cons(Y,Z)) -> nth(X,Z)
sumList(X) -> cons(X, sumList(add(X,s(X))))    sum(X) -> nth(X,sumList(X))

\Automaton A0
\States qnat qsum
\FinalStates qsum
\Transitions zero->qnat  s(qnat)->qnat  sum(qnat)->qsum

\Equations Simpl
\Rules
%Ec
cons(X, cons(Y, Z)) = cons(X, Z)              s(s(X))=s(X)

%E_R
add(zero,X)=X                                 nth(zero,cons(X,Y))=X
add(s(X),Y)=s(add(X,Y))                       nth(s(X),cons(Y,Z))=nth(X,Z)
sumList(X)=cons(X,sumList(add(X,s(X))))       sum(X)=nth(X,sumList(X))

%E^r
zero=zero               s(X)=s(X)             nth(X,Y)=nth(X,Y)
add(X,Y)=add(X,Y)       sum(X)=sum(X)         sumList(X)=sumList(X)
nil=nil                 cons(X,Y)=cons(X,Y)
\end{alltt}
}

\noindent
Using this specification \timbuk{} finds reachable irreducible terms, but \timbuk STRAT
succeeds in showing that the set of irreducible innermost reachable terms is
empty.

\section{Extension to leftmost and rightmost innermost strategy}
\label{right}
Real programming languages generally impose an additional strategy on the order
on which arguments at the same level are reduced, {\em e.g.} leftmost or
rightmost. For instance, the evaluation strategy of OCaml is rightmost
innermost. Not taking into account this additional requirement in the
reachability analysis may, again, lead to imprecise analysis in some particular
cases. This could be shown on an OCaml program but for sake of brevity we use a
simple TRS
\begin{exa} Let $R=\{a \rw b, c \rw c\}$. 
Assume that we use the rightmost innermost strategy. From the term
$f(a,c)$ it is not possible to reach the term $f(b,c)$ because rewriting $c$
does not terminate ({\em e.g.} a function call is looping), and rewriting $a$
will not be considered until $c$ is rewritten to a normal form (which is not
possible). 
\end{exa}

The innermost completion technique described above does not
take into account the order of evaluation of redexes at the same level. Thus,
$f(b,c)$ will be considered as reachable because $f(b,c)\in\Rin(\{f(a,c)\})$.
However, this can be improved. The innermost completion technique can be adapted
to tackle this problem, and correctness and precision theorems of
Sections~\ref{sect-th} and~\ref{sect-precision} can be lifted to take order of
evaluation into account. In the following, we instantiate this on the rightmost
innermost strategy but any other order (leftmost or even an order specific to each functional symbol) could be used. We denote by $\Rrin(u)$ the set of
terms reachable by one step of rightmost innermost rewriting of terms of $u$. To closely approximate $\Rrin$, the idea is simply to change the
way supplementary transitions are added by completion. Now, when solving a critical
pair $(\ell\to r, \sigma, q)$, for each transition   
\[f(\paire(q_1,p_1),\ldots, \paire(q_{i-1},p_{i-1}), \paire(q,p_{red}),
\paire(q_{i+1},p_{i+1}),\ldots,\paire(q_n,p_n)) \move 
\paire(q'',p'')\]
we add a supplementary transition \[f(\paire(q_1,p_1),\ldots, \paire(q_{i-1},p_{i-1}), \paire(q,p_{r\sigma}), \paire(q_{i+1},p'_{i+1}),\ldots,\paire(q_n,p'_n)) \move \paire(q'',p''')\]
only if there exists states $\paire(q_{i+1},p'_{i+1}),\ldots,\paire(q_n,p'_n)$
such that $p_j \neq \pred$ for $i+1 \leq j \leq n$\footnote{By choosing different
constraints on the $i's$ such that $p_i\neq \pred$, we can easily adapt this to
get leftmost innermost. And by associating this constraint to the symbol $f$ we
can cover the case of symbol specific normalizing strategy like in
context-sensitive rewriting.}
\begin{exa}[Continued]
On the above example, this would lead to the following completion. Assume that
the initial automaton contains transitions $a \move \paire(q_a,\pred)$, $c \move
\paire(q_c,\pred)$ and $f(\paire(q_a,\pred),\paire(q_c,\pred)) \move
\paire(q_f,\pred)$. Innermost completion would add the transitions $b \move
\paire(q_b,p_b)$ and $\paire(q_b,p_b) \movecol R \paire(q_a,p_b)$.
However, since the only transition having $\paire(q_a,\pred)$ as an argument is 
$f(\paire(q_a,\pred),\paire(q_c,\pred)) \move \paire(q_f,\pred)$ and since 
the only state associated to $q_c$ is $\pred$, we do not add supplementary
transitions. Completing the critical pair on $c \move
\paire(q_c,\pred)$ will not change the situation since it only yields the
transition $\paire(q_c,\pred) \movecol R \paire(q_c,\pred)$.
\end{exa}

\noindent
Now we show how to lift the correctness and precision theorems to the case of
rightmost innermost strategy. Since the only difference between the two
algorithms lies only in the way supplementary transitions are added, the proofs
are essentially the same. The only difference are in the lemmas where supplementary
transitions are considered. Thus, we focus on the differences w.r.t. the (general) innermost
case.

First, we define the notion of an automaton correct w.r.t. $\Rrin$. It is enough to
replace $\Rin$ by $\Rrin$ in Definition~\ref{defn-correct}.
Then we can prove the following lemma, equivalent to Lemma~\ref{lemme-correct}.

\begin{lem}\label{lemme-correct-right}\label{lemme-eqcomp-preserve-right}
 Any automaton produced by innermost completion starting from some $\autoInit\times\airr(R)$ is correct w.r.t. $\Rrin$.
\end{lem}
\begin{proof}
  The proof is similar to the general innermost case, except for what follows.
  Let $\paire(q,\pred)$ be a state of $\automate$ obtained by completion of
  $\autoInit\times\airr(R)$. Let $u\in\langrecs\automate{\paire(q,\pred)}$ and
  $v\in \Rrin(u)$.  By definition of rightmost innermost rewriting, there is a
  rule $\ell\to r$ of $R$, a substitution
  $\mu:\ensvar\to\langtermesclos\ensfonc$ and a context $C$ such that
  $u=C[\ell\mu]$, $v=C[r\mu]$ and each strict subterm of $\ell\mu$ is a normal
  form. If $C[\,]$ is not the empty context, using the rightmost hypothesis, we
  also know that there exists a context $C'[\,]$, a symbol $f$ of arity $n$ and
  terms $t_i$, $1\leq i \leq n$ such that $C[\ell\mu] =
  C'[f(t_1,\ldots,t_{i-1},\ell\mu,t_{i+1},\ldots,t_n)]$ and $t_{i+1},\ldots,t_n$
  are irreducible. Since,
  $C'[f(t_1,\ldots,t_{i-1},\ell\mu,t_{i+1},\ldots,t_n)]$ is recognized by $\A$
  there exists a transition $f(\paire(q_1,p_1), \ldots, \paire(q_n,p_n)) \move
  \paire(q',p')$ in $\A$. Then, since $t_{i-1},\ldots,t_n$ are irreducible and $\A$ is
  consistent with $\airr(R)$ (by Lemma~\ref{lemme-projdroite-irr}), we know that
  $p_{i+1},\ldots, p_n$ are all different from $\pred$. Thus there exists a
  supplementary transition
\[f(\paire(q_1,p_1), \ldots,
  \paire(q_{i-1},p_{i-1}),\paire(q_{\ell\sigma},p_{r\sigma}),
  \paire(q_{i+1},p'_{i+1}),\ldots,\paire(q_n,p'_n)) \move \paire(q',p'')\] 
in
  $\A$. The remainder of the proof can be carried out in the same way because
  this transition ensures that 
  $v=C'[f(t_1,\ldots,t_{i-1},r\sigma,t_{i+1},\ldots,t_n)] \moveet\A \paire(q,p''')$. 
\end{proof}

Then, using this lemma, the proof of the correctness theorem is an easy
adaptation of the initial proof.

\begin{thm}[Correctness]
\label{thm-correction-right}
 Assuming $R$ is left-linear, the innermost equational completion procedure,
 adapted for rightmost strategy, produces a correct result whenever it
 terminates and produces some fixpoint $\afixrin$ such that:
 \begin{equation}
  \langrec{\projgauche\afixrin}\supset \Rrin^*(\langrec{\autoInit}).
 \end{equation}

 \begin{equation}
  \langrec{\afixrin}\supset \Rrin^!(\langrec{\autoInit}).
 \end{equation}
\end{thm}
\begin{proof}
Similar to the general innermost case, except that the $\Rin$ correctness is
replaced by the $\Rrin$ correctness and Lemma~\ref{lemme-correct} is
replaced by Lemma~\ref{lemme-correct-right}.
\end{proof}

\noindent
For the precision theorem, we use the notion of $\Rrin/E$-coherence which is
defined as $\Rin/E$-coherence where $\Rrin$ replaces $\Rin$.
The precision theorem can seamlessly be extended to the rightmost innermost
case. This is due to the fact that, for a given state $\paire(q,p)$ on which a
completion step occurs, there is {\em no difference} between general and rightmost
innermost. The only difference is observed when building contexts above
$\paire(q,p)$ which is the role of supplementary transitions.
Supplementary transitions only appear in the proof of Lemma~\ref{lem-supop}. This
lemma has to be transformed and proved again. The statement and proof of all the
other lemmas of Section~\ref{sect-precision} can be used as they are, where
$\Rin/E$-coherence is replaced by $\Rrin/E$-coherence.  
Now let us focus on the unique lemma which has to be transformed.

\begin{lem}[Supplementary transitions preserve $\Rrin/E$-coherence]
\label{lem-supop-right}
Let $\A$ be a $\Rrin/E$-coherent automaton that is consistent with $\airr(R)$.
Assume that $\A$ has been completed with the transitions to have $s
\moveet\A \paire(q,p)$ and the necessary supplementary transitions.
Let $C[\,]$ be a context and $\paire(q_c,p_c)$ a
state of $\A$ such that $C[s]\moveet\A \paire(q_c,p_c)$. If 
$\langrecs\A{s}\subset(\Rrin/E)^*\left(\classeauto{\A\sseps} E
  {q}\right)$ then supplementary transitions ensure that 
$\langrecs\A{C[s]}\subset(\Rrin/E)^*\left(\classeauto{\A\sseps} E {q_c}\right)$.
\end{lem}
\begin{proof}
  Like in the proof of Lemma~\ref{lem-supop}, we make a proof by induction on
  the height of context $C[\,]$. The only small difference is in the inductive
  case when we know that $C[s] = C'[f(t_1,\ldots,t_{i-1},s,t_{i+1},\ldots,t_n)]$
  and 
  that there exists a supplementary transition
 \[f(\paire(q_1,p_1),\ldots,\paire(q_{i-1},p_{i-1}),\paire(q,p_s),\paire(q_{i+1},p_{i+1}),
  \ldots, \paire(q_n,p_n)) \move \paire(q',p')\] in $\A$ such that
  $C'[\paire(q',p')] \moveet\A \paire(q_c,p_c)$. For this transition to be
  added as a supplementary transition, we know that $p_{i+1},\ldots,p_n$ are all
  different from 
  $\pred$, and that there exists a transition
  \[f(\paire(q_1,p'_1),\ldots,\paire(q_{i-1},p'_{i-1}),\paire(q,p_{red}),\paire(q_{i+1},p'_{i+1}),
  \ldots, \paire(q_n,p'_n)) \move \paire(q',\pred)\] in the initial tree automaton
  which is $\Rrin/E$-coherent. Thus, initially, we trivially have
  $\langrecs\A{q'}\subset(\Rrin/E)^*\left(\classeauto{\A\sseps} E
    {q'}\right)$. Adding the transition where $\paire(q,p_s)$ replaces
  $\paire(q,p_{red})$ and $\paire(q_{i+1},p_{i+1}),\ldots \paire(q_n,p_n)$
  replaces $\paire(q_{i+1},p'_{i+1}),\ldots \paire(q_n,p'_n)$ also preserves
  this property on $q'$ because we know by assumption that
  $\langrecs\A{s}\subset(\Rrin/E)^*\left(\classeauto{\A\sseps} E
    {q}\right)$ and that terms in $\langrecs{\A}{\paire(q_j,p_j)}$ for $i+1 \leq
  j \leq n$ are irreducible (recall that $p_j\neq \pred$, for $i+1 \leq
  j \leq n$).
\end{proof}

With this new lemma, proving that completion steps preserve
$\Rrin/E$-coherence is done in the same way that in the general innermost case.
Finally, the precision theorem can be stated for the rightmost innermost case.
\begin{lem}\label{l1-right}
 Completion of an innermost critical pair, adapted for rightmost strategy
 preserves $\Rrin/E$-coherence. 
\end{lem}
\begin{proof}
  The proof still combines Lemmas~\ref{lemme-norm-pres}, \ref{lemme-heureka1}
  and uses Lemma~\ref{lem-supop-right} instead of Lemma~\ref{lem-supop}.
\end{proof}

\begin{thm}[Precision]\label{thm-precision-right}
 Let $E$ be a set of equations. Let $\A_0=\autoInit\times\airr(R)$, where
 $\autoInit$ has designated final states. We prune $\A_0$ of its non-accessible
 states. Suppose $\A_0$ separates the classes of $E$. Let $R$ be any left-linear
 TRS. Let $\A_i$ be obtained from $\A_0$ after some steps of innermost
 equational completion adapted for rightmost strategy. Then
  \begin{equation}
  \langrec{\projgauche{\A_i}}\subset (\Rrin/E)^*(\langrec{\autoInit})). 
 \end{equation}

  \begin{equation}
  \langrec{\A_i}\subset (\Rrin/E)^!(\langrec{\autoInit})).
 \end{equation}
\end{thm}

\begin{proof} 
  The proof is exactly the same as the one of the general innermost case, except that
  Lemma~\ref{l1-right} now replaces Lemma~\ref{l1}.
\end{proof}

\section{Related work} 
\label{related}
No tree automata completion-like
techniques~\cite{Genet-RTA98,TakaiKS-RTA00,
  BoichutCHK-IJFCS09,GenetR-JSC10,Lisitsa-RTA12} take evaluation strategies into
account. They compute over-approximations of {\em all} reachable terms.
Nevertheless, some of them~\cite{TakaiKS-RTA00, BoichutCHK-IJFCS09} can handle
non left-linear rules which are out of reach of the innermost completion
technique presented here. One of the main reason is that, if $R$ is non left-linear, the set $\lirr(R)$ may
not be regular. Thus the automaton $\airr(R)$ which is necessary to initiate the
innermost completion, may not exist. Some ways to overcome this limitation are
proposed in Section~\ref{perspectives}.

Dealing with reachable terms and strategies was first addressed
in~\cite{RetyV-RTA02} in the exact case for innermost and outermost strategies
but only for some restricted classes of TRSs, and also
in~\cite{GodoyJacquemard-WRS08}. As far as we know, the technique we propose is the
first to over-approximate terms reachable by innermost rewriting for {\em any}
left-linear TRSs. For instance, Example~\ref{ex:completionInnermost} and
examples of Section~\ref{sect-fun} and~\ref{equation} are in the scope of innermost equational completion but are outside
of the classes of~\cite{RetyV-RTA02,GodoyJacquemard-WRS08}. For instance, the {\tt
  sum} example is outside of classes of~\cite{RetyV-RTA02,GodoyJacquemard-WRS08} because a right-hand side
of a rule has two nested defined symbols and is not shallow.

Data flow analysis of higher-order functional programs is a long standing
and very active research topic~\cite{OngR-POPL11,KochemsO-RTA11,JonesA-TCS07}.
Used techniques range from tree grammars to specific formalisms: HORS, PMRS or
ILTGs and can deal with higher-order functions. We have shown, on an example,
that defining an analysis taking the call-by-value evaluation strategy was also
possible on higher-order functions. However, this has to be investigated more
deeply. Application of innermost completion to higher order function would
provide nice improvements on static analysis techniques. Indeed, state of the
art techniques like~\cite{OngR-POPL11,KochemsO-RTA11,JonesA-TCS07} do not take
evaluation strategies into account,  and analysis results are thus 
coarse when program execution relies on a specific strategy. 

A recent paper~\cite{CirsteaLM-RTA15} shows how to encode a strategy into a
TRS. An attractive alternative to the work presented here is to use standard
completion on the encoding of innermost strategy on a given TRS $R$. 
We experimented with this technique to see if it was possible to enlarge
the family of strategies that completion can deal with. However, this raises several
problems. First, the transformed TRSs are huge and complex. For instance, the
transformed TRS encoding the {\tt sum} example under innermost strategy consists
of: 63 symbols, 706 variables and 620 rules. Completion of this TRS is far more
costly than innermost completion of the initial {\tt sum} TRS that has only 6
rules. Second, and more critical, on the transformed TRS the termination of
completion is impossible to guarantee. On the {\tt sum} example we shown in
Section~\ref{equation} how equations can be generated from the TRS using
termination results of~\cite{Genet-JLAMP15}. This is not possible on the
transformed TRS because it does not conform to the typed functional TRS
schema required by~\cite{Genet-JLAMP15}. To have a termination guarantee on the
transformed TRS one would need an extension of the termination results
of~\cite{Genet-JLAMP15} on general TRS. This extension is ongoing work.

\section{Conclusion}

In this paper, we have proposed a sound and precise algorithm over-approximating
the set of terms reachable by innermost rewriting. As far as we know this is the
first algorithm solving this problem for any left linear TRS and any regular
initial set of terms. It is based on tree automata completion and equational
abstractions with a set $E$ of approximation equations. The algorithm also
minimizes the set of added transitions by completing the product automaton
(between $\autoInit$ and $\airr(R)$). We proposed \timbuk STRAT~\cite{timbuk}, a prototype implementation of this method.

The precision of the approximations have been shown on a theoretical and a
practical point of view. On a theoretical point of view, we have shown that the
approximation automaton recognises no more terms than those effectively
reachable by innermost rewriting modulo the approximation $E$. On the practical
side, unlike other techniques used to statically analyze functional
programs~\cite{OngR-POPL11,KochemsO-RTA11,JonesA-TCS07},
innermost equational completion can take the call-by-value strategy into account. As a result, for
programs whose semantics highly depend on the evaluation strategy, innermost
equational completion yields more accurate results. This should open new ways to
statically analyze functional programs by taking evaluation strategies into
account. 

Approximations of sets of ancestors or descendants can also improve existing
termination techniques~\cite{GeserHWZ-RTA05,Middeldorp-ENTCS02}.
In the dependency pairs setting, such approximations can remove edges in a
dependency graph by showing that there is no rewrite derivation from a
pair to another. Besides, it has been shown that dependency pairs can prove innermost
termination~\cite{GieslT-JAR06}. In this case, {\em innermost} equational completion can
more strongly prune the dependency graph: it can show that there is no {\em
  innermost} derivation from a pair to another. For instance, on the TRS:

{\small
$\begin{array}{l|l|l}
choice(X,Y) \rw X   &   choice(X,Y) \rw Y    &   eq(s(X),s(Y)) \rw eq(X,Y) \\
eq(0,0) \rw \TT      & eq(s(X),0) \rw \FF      & eq(0,s(Y)) \rw \FF \\
g(0,X) \rw eq(X,X)  & g(s(X),Y) \rw g(X,Y)   &  f(\FF,X,Y) \rw f(g(X,choice(X,Y)),X,Y) 
\end{array}$}

\noindent
We can prove that any term of the form $f(g(t_1,choice(t_2,t_3)),t_4,t_5)$
cannot be rewritten (innermost) to a term of the form $f(\FF,t_6,t_7)$ (for all
terms $t_i\in\TF$, $i=1\ldots 7$). This proves that, in the dependency graph,
there is no cycle on this pair. This makes the termination proof of this TRS
simpler than what AProVE~\cite{GieslAPROVE-IJCAR14} does: it needs more complex
techniques, including proofs by induction. Simplification of termination proofs
using innermost equational completion should be investigated more deeply. 

\section{Perspectives}
\label{perspectives}
For further work, we want to improve and expand our implementation of innermost
equational completion in order to design a strategy-aware and higher-order-able
static analyzer for a reasonable subset of a real functional programming
language with call-by-value like OCaml, F\#, Scala, Lisp or Scheme. To translate
OCaml programs to TRS, a possible solution would be to use the translator of
HOCA~\cite{AvanziniLM-ICFP15}. HOCA translates a subset of higher-order OCaml
programs to TRS to perform complexity analysis. HOCA uses the same
encoding as the one we used in Section~\ref{sect-fun}. We already showed
in~\cite{GenetS-rep13} that completion can perform static analysis on examples
taken from~\cite{OngR-POPL11}. We also want to study if the
innermost completion covers the TRS classes preserving regularity
of~\cite{RetyV-RTA02,GodoyJacquemard-WRS08}. Note that 
Corollary~\ref{corollary-emptyset} already ensures that if completion
terminates with $E=\emptyset$ then it exactly computes innermost reachable
terms. Thus, proving that it covers the classes
of~\cite{RetyV-RTA02,GodoyJacquemard-WRS08} would essentially consist in proving
termination of innermost completion on those classes and with $E=\emptyset$. A
similar proof technique has already been used for standard completion and general
rewriting~\cite{FeuilladeGVTT-JAR04,Genet-JLAMP15}.

As explained in Section~\ref{related}, innermost completion cannot handle non
left-linear TRS because the set of irreducible terms may not be recognized by a
tree automaton. A possible research direction would be to
replace, in innermost completion, $\airr(R)$ by a {\em deterministic reduction
  automaton} (see Section~4.4.5 of~\cite{tata}) recognizing $\lirr(R)$. However,
most of the algorithms of the reduction automaton class, needed in completion,
have a very high complexity. A simple workaround would be, instead, to use a tree
automaton $\airr^+(R)$ over-approximating $\lirr(R)$. This would trigger more critical pairs and thus produce a bigger
(though correct) over-approximation. Then, if testing the reducibility can be solved
in an exact or approximated way, dealing with non left-linear rules in
completion may be easier in the innermost than in the general case.  Roughly, to
solve a critical pair between a non left-linear rule $f(x,x) \rw g(x)$ and a
tree automaton transition of the form $f(q_1,q_2) \move q$ it is necessary to
check whether there exist terms recognized by $q_1$ and $q_2$. This test is
necessary because tree automata produced by completion are not deterministic in
general. Then, if the test is true, completion adds a transition of the form
$g(q_3) \move q$ and {\em all the necessary transitions} to have
$\langrecs{\A}{q_3} = \langrecs{\A}{q_1}\cap \langrecs{\A}{q_2}$. Thus, with
non left-linear rules, critical pair solving (and detection) becomes more complex
and may result into a huge number of new transitions (the completed automaton may exponentially
grow-up w.r.t. the number of completion steps). Both~\cite{TakaiKS-RTA00, BoichutCHK-IJFCS09} propose sophisticated
techniques and data structures to limit the blow-up in practice. Surprisingly,
in the innermost case, the situation is likely to be a little bit more
favorable. To check if there is a critical pair between the non left-linear rule
$f(x,x) \rw g(x)$ and a transition of the form
$f(\paire(q_1,p_1),\paire(q_2,p_2)) \move \paire(q,p)$, we know that the
languages recognized by $\paire(q_1,p_1)$ and $\paire(q_2,p_2)$ consist only of
irreducible terms. If the set of transitions recognizing irreducible terms is
deterministic, then to decide if the language is empty we can check if
$\paire(q_1,p_1)$ is {\em syntactically} equal to
$\paire(q_2,p_2)$. Furthermore, to solve the critical pair, it is enough to add
the {\em two} transitions $g(\paire(q_1,p_1)) \move \paire(q_3,p_3)$ and
$\paire(q_3,p_3) \move \paire(q,p_3)$.  For functional TRS~\cite{Genet-JLAMP15}
and for constructor TRS~\cite{Rety-LPAR99}, the set of irreducible terms is
known a priori: they are constructor terms, {\em i.e.} the data terms of the TRS.
Thus the set of transitions recognizing those constructor terms can be defined
in deterministic way and will not be modified by completion. In particular, 
it will remain deterministic during the
completion process\footnote{More precisely, completion may add some epsilon transitions
  in this set of transitions, but the syntactical equality test can be replaced 
  by a test modulo epsilon transition closure.}. We need to check if this makes
it possible to {\em efficiently} approximate innermost reachable terms in the
presence of non left-linear rules.

Another objective is to extend this completion technique to other strategies.
Another strategy of interest for completion is the outermost strategy. This
would improve the precision of static analysis of functional programming
language using call-by-need evaluation strategy, like Haskell.  Extension of
this work to the outermost case is not straightforward but it may use similar
principles, such as running completion on a pair automaton rather than on single
automaton.  States in tree automata are closely related to positions in
terms. To deal with the innermost strategy, in states $\paire(q,p)$, the $p$
component tells us if terms $s$ (or subterms of $s$) recognised by the state
$\paire(q,p)$ are reducible or not. This is handy for innermost completion
because we can decide if a tuple $(\ell \to r, \sigma, \paire(q',p'))$ is an
{\em innermost} critical pair by checking if the $p$ components of the states
recognising strict subterms of $\ell\sigma$ are different from $p_{red}$. For
the outermost case, this is exactly the opposite: a tuple $(\ell \to r, \sigma,
\paire(q',p'))$ is an {\em outermost} critical pair only if all the {\em
  contexts} $C[\,]$ such that $C[\ell\sigma]$ is recognised, are irreducible
contexts. If it is possible to encode in the $p'$ component (using an automaton
or something else) whether all contexts embedding $\paire(q',p')$ are
irreducible or not, we should be able to define outermost critical pairs and,
thus, outermost completion in a similar manner.

\section*{Acknowledgments} The authors thank Ren\'{e} Thiemann for providing
the example of 
innermost terminating TRS for AProVE, Thomas Jensen, Luke Ong, Jonathan Kochems,
Robin Neatherway and the anonymous referees for their valuable comments and suggestions.

\bibliographystyle{plain}
\bibliography{sabbrev,genet,eureca}

\begin{thebibliography}{10}

\bibitem{AvanziniLM-ICFP15}
M.~Avanzini, U.~Dal~Lago, and G.~Moser.
\newblock Analysing the complexity of functional programs: higher-order meets
  first-order.
\newblock In {\em ICFP'15}, pages 152--164. {ACM}, 2015.

\bibitem{BaaderN-book98}
F.~Baader and T.~Nipkow.
\newblock {\em Term Rewriting and All That}.
\newblock Cambridge University Press, 1998.

\bibitem{BoichutCR-RTA13}
Y.~Boichut, J.~Chabin, and P.~R{\'{e}}ty.
\newblock Over-approximating descendants by synchronized tree languages.
\newblock In {\em RTA'13}, volume~21 of {\em LIPIcs}, pages 128--142. Schloss
  Dagstuhl - Leibniz-Zentrum fuer Informatik, 2013.

\bibitem{BoichutCHK-IJFCS09}
Y.~Boichut, R.~Courbis, P.-C. H{\'e}am, and O.~Kouchnarenko.
\newblock Handling non left-linear rules when completing tree automata.
\newblock {\em IJFCS}, 20(5), 2009.

\bibitem{BroadbentCHS-ICFP13}
C.~H. Broadbent, A.~Carayol, M.~Hague, and O.~Serre.
\newblock C-shore: a collapsible approach to higher-order verification.
\newblock In {\em ICFP'13}. {ACM}, 2013.

\bibitem{CastagnaNXHLP-POPL14}
G.~Castagna, K.~Nguyen, Z.~Xu, H.~Im, S.~Lenglet, and L.~Padovani.
\newblock Polymorphic functions with set-theoretic types: part 1: syntax,
  semantics, and evaluation.
\newblock In {\em POPL'14}. {ACM}, 2014.

\bibitem{CirsteaLM-RTA15}
H.~Cirstea, S.~Lenglet, and P.-E. Moreau.
\newblock A faithful encoding of programmable strategies into term rewriting
  systems.
\newblock In {\em RTA'15}, volume~36 of {\em LIPIcs}, pages 74--88. Schloss
  Dagstuhl - Leibniz-Zentrum fuer Informatik, 2015.

\bibitem{comon-IC00}
H.~Comon.
\newblock Sequentiality, {M}onadic {S}econd-{O}rder {L}ogic and {T}ree
  {A}utomata.
\newblock {\em Inf. Comput.}, 157(1-2):25--51, 2000.

\bibitem{tata}
H.~Comon, M.~Dauchet, R.~Gilleron, F.~Jacquemard, D.~Lugiez, C.~L\"{o}ding,
  S.~Tison, and M.~Tommasi.
\newblock Tree automata techniques and applications.
\newblock {\small {\tt http://tata.gforge.inria.fr}}, 2008.

\bibitem{RemyComon87}
H.~Comon and Jean-Luc R{\'e}my.
\newblock How to characterize the language of ground normal forms.
\newblock Technical Report 676, INRIA-Lorraine, 1987.

\bibitem{FeuilladeGVTT-JAR04}
G.~Feuillade, T.~Genet, and V.~Viet Triem~Tong.
\newblock {R}eachability {A}nalysis over {T}erm {R}ewriting {S}ystems.
\newblock {\em Journal of Automated Reasonning}, 33 (3-4):341--383, 2004.

\bibitem{GodoyJacquemard-WRS08}
A.~Gascon, G.~Godoy, and F.~Jacquemard.
\newblock Closure of {T}ree {A}utomata {L}anguages under {I}nnermost
  {R}ewriting.
\newblock In {\em WRS'08}, volume 237 of {\em ENTCS}, pages 23--38. Elsevier,
  2008.

\bibitem{Genet-RTA98}
T.~Genet.
\newblock Decidable {A}pproximations of {S}ets of {D}escendants and {S}ets of
  {N}ormal {F}orms.
\newblock In {\em RTA'98}, volume 1379 of {\em LNCS}, pages 151--165. Springer,
  1998.

\bibitem{Genet-rep14}
T.~Genet.
\newblock A note on the {P}recision of the {T}ree {A}utomata {C}ompletion.
\newblock Technical report, INRIA, 2014.
\newblock \url{https://hal.inria.fr/hal-01091393}.

\bibitem{Genet-WRLA14}
T.~Genet.
\newblock Towards {S}tatic {A}nalysis of {F}unctional {P}rograms using {T}ree
  {A}utomata {C}ompletion.
\newblock In {\em WRLA'14}, volume 8663 of {\em LNCS}. Springer, 2014.

\bibitem{Genet-JLAMP15}
T.~Genet.
\newblock Termination {C}riteria for {T}ree {A}utomata {C}ompletion.
\newblock {\em Journal of {L}ogical and {A}lgebraic {M}ethods in
  {P}rogramming}, 85, Issue 1, Part 1:3--33, 2016.

\bibitem{timbuk}
T.~Genet, Y.~Boichut, B.~Boyer, V.~Murat, and Y.~Salmon.
\newblock {R}eachability {A}nalysis and {T}ree {A}utomata {C}alculations.
\newblock IRISA / Universit\'e de Rennes 1.
\newblock \url{http://www.irisa.fr/celtique/genet/timbuk/}.

\bibitem{GenetR-JSC10}
T.~Genet and R.~Rusu.
\newblock Equational tree automata completion.
\newblock {\em Journal of {S}ymbolic {C}omputation}, 45:574--597, 2010.

\bibitem{GenetS-rep13}
T.~Genet and Y.~Salmon.
\newblock {Tree Automata Completion for Static Analysis of Functional
  Programs}.
\newblock Technical report, INRIA, 2013.
\newblock \url{http://hal.archives-ouvertes.fr/hal-00780124/PDF/main.pdf}.

\bibitem{GenetS-RTA15}
T.~Genet and Y.~Salmon.
\newblock Reachability {A}nalysis of {I}nnermost {R}ewriting.
\newblock In {\em RTA'15}, volume~36 of {\em LIPIcs}, Warshaw, 2015. Schloss
  Dagstuhl - Leibniz-Zentrum fuer Informatik.

\bibitem{GeserHWZ-RTA05}
A.~Geser, D.~Hofbauer, J.~Waldmann, and H.~Zantema.
\newblock On tree automata that certify termination of left-linear term
  rewriting systems.
\newblock In {\em RTA'05}, volume 3467 of {\em LNCS}, pages 353--367. Springer,
  2005.

\bibitem{GieslAPROVE-IJCAR14}
J.~Giesl, M.~Brockschmidt, F.~Emmes, F.~Frohn, C.~Fuhs, C.~Otto,
  M.~Pl{\"{u}}cker, P.~Schneider{-}Kamp, T.~Str{\"{o}}der, S.~Swiderski, and
  R.~Thiemann.
\newblock Proving termination of programs automatically with aprove.
\newblock In {\em IJCAR'14}, volume 8562 of {\em LNCS}, pages 184--191.
  Springer, 2014.

\bibitem{GieslT-JAR06}
J.~Giesl, R.~Thiemann, P.~Schneider-Kamp, and S.~Falke.
\newblock Mechanizing and improving dependency pairs.
\newblock {\em Journal of Automated Reasonning}, 37(3):155--203, 2006.

\bibitem{JonesA-TCS07}
N.~D. Jones and N.~Andersen.
\newblock Flow analysis of lazy higher-order functional programs.
\newblock {\em Theoretical Computer Science}, 375(1-3):120--136, 2007.

\bibitem{Kobayashi-ACM13}
N.~Kobayashi.
\newblock {M}odel {C}hecking {H}igher-{O}rder {P}rograms.
\newblock {\em Journal of the ACM}, 60.3(20), 2013.

\bibitem{KobayashiI-ESOP13}
N.~Kobayashi and A.~Igarashi.
\newblock Model-{C}hecking {H}igher-{O}rder {P}rograms with {R}ecursive
  {T}ypes.
\newblock In {\em ESOP'13}, volume 7792 of {\em LNCS}, pages 431--450.
  Springer, 2013.

\bibitem{KobayashiTU-POPL10}
N.~Kobayashi, N.~Tabuchi, and H.~Unno.
\newblock Higher-order multi-parameter tree transducers and recursion schemes
  for program verification.
\newblock In {\em POPL'10}, pages 495--508. ACM, 2010.

\bibitem{KochemsO-RTA11}
J.~Kochems and L.~Ong.
\newblock {Improved Functional Flow and Reachability Analyses Using Indexed
  Linear Tree Grammars}.
\newblock In {\em RTA'11}, volume~10 of {\em LIPIcs}. Schloss
  Dagstuhl--Leibniz-Zentrum fuer Informatik, 2011.

\bibitem{Lisitsa-RTA12}
A.~Lisitsa.
\newblock {Finite Models vs Tree Automata in Safety Verification}.
\newblock In {\em RTA'12}, volume~15 of {\em LIPIcs}, pages 225--239, 2012.

\bibitem{Middeldorp-ENTCS02}
A.~Middeldorp.
\newblock Approximations for strategies and termination.
\newblock {\em ENTCS}, 70(6):1--20, 2002.

\bibitem{OngR-POPL11}
L.~Ong and S.~Ramsay.
\newblock Verifying higher-order functional programs with pattern-matching
  algebraic data types.
\newblock In {\em POPL'11}. {ACM}, 2011.

\bibitem{Rety-LPAR99}
P.~R\'{e}ty.
\newblock {R}egular {S}ets of {D}escendants for {C}onstructor-based {R}ewrite
  {S}ystems.
\newblock In {\em Proc.\ 6th LPAR Conf., Tbilisi (Georgia)}, volume 1705 of
  {\em LNAI}. Springer-Verlag, 1999.

\bibitem{RetyV-RTA02}
P.~R\'ety and J.~Vuotto.
\newblock {R}egular {S}ets of {D}escendants by some {R}ewrite {S}trategies.
\newblock In {\em RTA'02}, volume 2378 of {\em LNCS}. Springer, 2002.

\bibitem{Reynolds-IP69}
J.~Reynolds.
\newblock Automatic computation of data set definitions.
\newblock {\em Information Processing}, 68:456--461, 1969.

\bibitem{Takai-RTA04}
T.~Takai.
\newblock {A} {V}erification {T}echnique {U}sing {T}erm {R}ewriting {S}ystems
  and {A}bstract {I}nterpretation.
\newblock In {\em RTA'04}, volume 3091 of {\em LNCS}, pages 119--133. Springer,
  2004.

\bibitem{TakaiKS-RTA00}
T.~Takai, Y.~Kaji, and H.~Seki.
\newblock {R}ight-linear finite-path overlapping term rewriting systems
  effectively preserve recognizability.
\newblock In {\em RTA'11}, volume 1833 of {\em LNCS}. Springer, 2000.

\bibitem{Terese}
Terese.
\newblock {\em Term Rewriting Systems}.
\newblock Cambridge University Press, 2003.

\bibitem{VazouRJ-ESOP2013}
N.~Vazou, P.~Rondon, and R.~Jhala.
\newblock Abstract {R}efinement {T}ypes.
\newblock In {\em ESOP'13}, volume 7792 of {\em LNCS}. Springer, 2013.

\end{thebibliography}

\end{document}